\documentclass{article}

\usepackage{arxiv}

\usepackage[utf8]{inputenc} 
\usepackage[T1]{fontenc}    
\usepackage{hyperref}       
\usepackage{url}            
\usepackage{booktabs}       
\usepackage{amsfonts}       
\usepackage{nicefrac}       
\usepackage{microtype}      
\usepackage{lipsum}		
\usepackage{graphicx}
\usepackage{natbib}
\usepackage{doi}
\usepackage{bm}
\usepackage{amsmath}

\usepackage{amsthm}
\newtheorem{theorem}{Theorem}

\title{LEAP: The latent exchangeability prior for borrowing information from historical data}

\date{March 7, 2023}	

\author{ \href{https://orcid.org/0000-0002-6112-9030}{\includegraphics[scale=0.06]{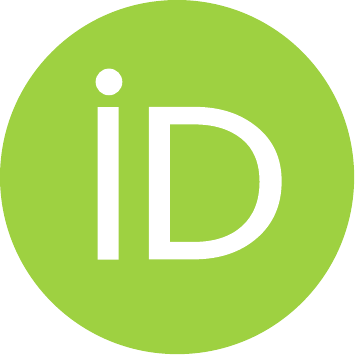}\hspace{1mm}Ethan M.~Alt} \\
	Department of Biostatistics\\
	University of North Carolina\\
	Chapel Hill, NC 27516 \\
	\texttt{ethanalt@live.unc.edu} \\
	\And
	\hspace{1mm}Xiuya Chang \\
	Department of Biostatistics\\
	University of North Carolina\\
	Chapel Hill, NC 27516 \\
	\texttt{coco.xyc@unc.edu} \\
	\And
	Xun Jiang \\
	Amgen \\
        Thousand Oaks, CA 91320 \\
	\texttt{xunj@amgen.com} \\
	\And
	Qing Liu \\
	Amgen \\
        Thousand Oaks, CA 91320 \\
	\texttt{qliu02@amgen.com} \\
	\And
	May Mo \\
	Amgen \\
        Thousand Oaks, CA 91320 \\
	\texttt{mm0@amgen.com} \\
 	\And
	H. Amy Xia \\
	Amgen \\
        Thousand Oaks, CA 91320 \\
	\texttt{mm0@amgen.com} \\
        \And
        Joseph G. Ibrahim \\
	University of North Carolina\\
	Chapel Hill, NC 27516 \\
	\texttt{ibrahim@bios.unc.edu} \\
}

\renewcommand{\shorttitle}{\textit{arXiv} Template}

\hypersetup{
pdftitle={A template for the arxiv style},
pdfsubject={q-bio.NC, q-bio.QM},
pdfauthor={David S.~Hippocampus, Elias D.~Striatum},
pdfkeywords={Bayesian dynamic borrowing, External data, Historical data, Clinical trials}
}

\begin{document}
\maketitle

\begin{abstract}
    It is becoming increasingly popular to elicit informative priors on the basis of historical data. Popular existing priors, including the power prior, commensurate prior, and robust meta-analytic prior provide blanket discounting. Thus, if only a subset of participants in the historical data are exchangeable with the current data, these priors may not be appropriate. In order to combat this issue, propensity score (PS) approaches have been proposed. However, PS approaches are only concerned with the covariate distribution, whereas exchangeability is typically assessed with parameters pertaining to the outcome. In this paper, we introduce the latent exchangeability prior (LEAP), where observations in the historical data are classified into exchangeable and non-exchangeable groups. The LEAP discounts the historical data by identifying the most relevant subjects from the historical data. We compare our proposed approach against alternative approaches in simulations and present a case study using our proposed prior to augment a control arm in a phase 3 clinical trial in plaque psoriasis with an unbalanced randomization scheme.
\end{abstract}


\section{Introduction}
\label{sec:intro}
It is becoming increasingly common to utilize statistical methods that incorporate prior information. Such prior information, such as historical data, is naturally integrated in the Bayesian paradigm, where analysts may specify a prior distribution for what evidence they possess prior to the analysis of the data. 
Bayesian methods for incorporating historical data have been used or proposed in wide a variety of statistical applications including, but not limited to, epidemiology \citep{warasi_estimating_2016}, political science \citep{isakov_towards_2020}, engineering \citep{lorencin_evaluating_2017}, spatial applications \citep{louzada_spatial_2021}, small area estimation \citep{young_using_2022}, and psychology \citep{konig_moving_2021}. 

In the areas of medical device and drug development, where historical data from clinical studies are abundant, it is of interest to incorporate historical data to improve a current trial's efficiency. For example, in pediatric studies, it can be difficult to recruit an appropriate number of patients, necessitating the use of incorporating prior information \citep{azzolina_handling_2021}. In rare diseases, prevalence is low by definition, and some have recommended Bayesian methods to borrow information from controls to reduce recruitment burden (e.g., \cite{chow_innovative_2020}). In disease areas with unmet need, the incorporation of historical data can make clinical trials more ethical by assigning fewer patients to the control arm. Conversely, in disease areas where an established standard of care exists with a wealth of historical data, conducting balanced randomization can add to the time and cost of patient recruitment, delaying the drug approval process and potentially costing lives. However, some or all subjects in the historical data might be different than the current data set, necessitating so-called dynamic borrowing techniques.

The literature on informative priors on the basis of historical data is rich. One of the earliest and most popular methods for leveraging historical data is the power prior (PP) of \cite{ibrahim_power_2000}. The PP involves raising the likelihood of the historical data to a power $a_0 \in [0, 1]$, which is sometimes referred to as the discounting parameter. The PP assumes that the parameters of the current and historical data sets are the same and allows the user to elicit a value $a_0$ to reflect the degree to which they may be different. A second option is the commensurate prior (CP) of \cite{hobbs_commensurate_2012}. Unlike the PP, the CP assumes that the parameters between the data sets could possibly be different, but that the parameters for the current data set are normally distributed with mean equal to the parameters of the historical data set. While the PP only permits blanket discounting, the CP allows parameter-specific discounting. A third popular approach is the robust meta-analytic predictive prior (RMAPP) of \cite{schmidli_robust_2014}. The RMAPP is essentially a mixture prior whose first component is an informative prior consisting of a meta-analytic predictive prior for the current data parameter and whose second component is a noninformative (vague) prior. Although the RMAPP is conceptually simple, it has only been developed for single parameter settings, limiting its applicability. Also, the RMAPP is not available in closed-form in general, so that approximations to the prior must be used in practice.

While the PP, CP, and RMAPP allow analysts to flexibly specify the amount of borrowing, they do not directly handle the case where perhaps only a subset of individuals in the historical data set are exchangeable. To this end, propensity score (PS) methods have been developed, where the probability of being enrolled in the current study is modeled as a function of baseline covariates to construct a PS. The PS serves as a measure of similarity between participants in the current and external data sets such that individuals with similar PSs are considered to be exchangeable. The end goal of PS methods is to end up with balance in baseline covariates. However, covariate balance does not ensure that individuals are exchangeable, as the outcome model could be different between the historical and current data sources, particularly when there is heterogeneity in study design.

In this paper, we introduce the latent exchangeability prior (LEAP). The LEAP is a simple, dynamic borrowing prior for historical data that is based on the assumption that possibly only a fraction of individuals in the historical data set are exchangeable with the current data set. The LEAP is a fully Bayesian characterization of exchangeability, providing an advantage over PS based approaches, for which the amount of borrowing necessarily depends on the baseline covariates of the current study. Moreover, the LEAP provides discounting at the individual level, permitting efficiency gains when a subset of patients in the historical data are exchangeable with subjects in the current data. This contrasts with the PP, CP, and RMAPP approaches, which only permit blanket discounting.

This paper proceeds as follows. In Section \ref{sec:motivation}, we motivate our prior with a placebo-controlled clinical trial with an unbalanced randomization scheme, where we wish to augment the control arm for efficiency gains. Section \ref{sec:leap_general} includes the development of the LEAP in full generality. We compare the LEAP to other approaches in Section \ref{sec:comparison}. Section \ref{sec:sims} reports results from extensive simulations for the normal linear model, comparing with other suggested priors. Section \ref{sec:analysis} provides results from data analysis using a real clinical trial. We conclude our results in Section \ref{sec:conc}.

\section{The ESTEEM Trials}
\label{sec:motivation}
Our proposed method is motivated by the ESTEEM trials for patients with moderate-to-severe plaque psoriasis. There was an unmet medical need for this patients population because many of them discontinued the conventional treatments due to safety issues and lack of tolerability. Apremilast, an oral phosphodiesterase 4 inhibitor regulates immune response associated with psoriasis, provided a novel therapeutic option for this population. 

The ESTEEM I trial \citep{papp_apremilast_2015} was the first phase 3 study evaluating the efficacy and safety of Apremilast for patients with moderate-to-severe plaque psoriasis, which was a double-blinded, placebo controlled randomized trial with 562 subjects randomized to Apremilast group and 282 subjects to placebo group. 
The ESTEEM II trial \citep{paul_efficacy_2015} was another phase 3 placebo controlled clinical study evaluating Apremilast for the same population but with smaller sample size. A total of 411 patients enrolled in the trial, and 274 subjects were randomized to Apremilast arm and 137 were assigned to placebo. We discuss the ESTEEM trials in more detail in Section \ref{sec:sims}.

Due to the unbalanced randomization, there is a potential for efficiency loss in the ESTEEM II trial. Since the ESTEEM I trial was a larger study with similar design, precision can be increased by borrowing information from the ESTEEM I study. To avoid bias and inflation of the type I error rate, only the most relevant patients from the study should be included. Thus, the goal is to improve the precision of treatment effect estimate for the ESTEEM II study by augmenting its control arm through a dynamic borrowing procedure.





\section{The latent exchangeabilty prior}
\label{sec:leap_general}
In this section, we develop the LEAP as motivated by the ESTEEM trials. We discuss prior elicitation and posterior propriety. We also introduce a concept called sample size contribution (SSC), which can be used to tune the amount of information borrowing.

\subsection{The general methodology}
Suppose that we possess historical data $D_0$ with $n_0$ observations. The LEAP assumes that the historical data come from independent observations of a mixture model with $K$ components, where the $k^{th}$ component is parameterized by $\bm{\theta}_k$. Let $\bm{\theta} = (\bm{\theta}_1', \ldots, \bm{\theta}_K')'$. In this article, we use the term ``density'' to refer generally to probability density functions (PDFs) for continuous random variables or probability mass functions (PMFs) for discrete random variables. The joint density for the historical data is given by
\begin{equation}
  f(D_0 | \bm{\theta}, \bm{\gamma}) = \prod_{i=1}^{n_0} \sum_{k=1}^K \gamma_k f(D_{0i} | \bm{\theta}_k),
  \label{eq:mixdens}
\end{equation}
where $D_{0i}$ is the $i^{th}$ observation of the historical data, $\bm{\gamma} = (\gamma_1, \ldots, \gamma_K)'$, and $\sum_{k=1}^K \gamma_k = 1$. Throughout this paper, we assume without loss of generality that the first component of the mixture in \eqref{eq:mixdens} is the same density function as the current data $D$. Thus, the parameter $\gamma_1$ quantifies the marginal probability that an individual in the historical data set is exchangeable with the current data set.

It is more convenient to work with the density in \eqref{eq:mixdens} using a latent variable representation. Specifically, let $c_{0i} \in \{1, \ldots, K\}$ denote the class to which the $i^{th}$ subject of the historical data set belongs, where we assume $c_{0i} \sim \text{Categorical}(\bm{\gamma})$ are i.i.d., i.e., the joint mass function of $\bm{c}_0 = (c_{01}, \ldots, c_{0,n_0})'$ is given by
$
  f(\bm{c}_0 | \bm{\gamma}) = \prod_{i=1}^{n_0} \prod_{k=1}^K \gamma_k^{c_{0ik}},
$
where $c_{0ik} = 1\{ c_{0i} = k \}$ is an indicator that equals 1 if subject $i$ in the historical data belongs to class $k$ and 0 otherwise. Then the joint density of $(D_0, \bm{c}_0)$ is given by
\begin{equation}
    f(D_0, \bm{c}_0 | \bm{\theta}, \bm{\gamma}) = \prod_{i=1}^{n_0} \prod_{k=1}^K \left[ \gamma_k f(D_{0i} | \bm{\theta}_k) \right]^{ c_{0ik}}.
    \label{eq:mixdens_latent}
\end{equation}
Note that when $c_{0i1} = 1$, the individual in the historical data set is exchangeable with the current data set. This motivates the namesake for the LEAP, namely, that latent variables indicate whether a given individual in the historical data set is exchangeable.

Let $\pi_0(\bm{\theta}, \bm{\gamma})$ denote a prior for $\bm{\theta}$ and $\bm{\gamma}$ that, adopting the terminology of \citet{ibrahim_power_2000} we refer to as the ``initial prior.'' The LEAP is defined as the prior induced by the density in \eqref{eq:mixdens_latent} and the initial prior. Mathematically, we may write the joint LEAP as
\begin{align}
    \pi(\bm{\theta}, \bm{\gamma}, \bm{c}_0 | D_0) &\propto f(D_0, \bm{c}_0 | \bm{\theta}, \bm{\gamma}) \pi_0(\bm{\theta}, \bm{\gamma}).
    \label{eq:jointleap}
\end{align}
The prior in \eqref{eq:jointleap} may be interpreted as the posterior density of the historical data, which is modeled as a finite mixture, with prior $\pi_0(\bm{\theta}, \bm{\gamma})$. The marginal LEAP is given by
\begin{align}
  \pi(\bm{\theta}_1 | D_0) = \int \cdots \int \sum_{\bm{c}_0 \in \{ 1, \ldots, K \}^{n_0}} \pi(\bm{\theta}, \bm{\gamma}, \bm{c}_0 | D_0) d\bm{\theta}_2 \cdots d\bm{\theta}_K d\bm{\gamma}.
  \label{eq:margleap}
\end{align}
Note that the marginal prior in \eqref{eq:margleap} marginalizes over our uncertainty about to which class each observation in the historical data belongs. The LEAP is thus a ``dynamic borrowing'' prior. Perhaps more importantly, the LEAP is applicable even if only a proportion of individuals in the historical data set are exchangeable.

Let $\pi(\bm{c}_0 | D_0) \propto \int \int \left[ \prod_{k=1}^K \gamma_k^{c_{0i}} \prod_{i=1}^{n_0} f\left( D_{0i} | \bm{\theta}_k \right) \right] d\bm{\theta} d\bm{\gamma} $ denote the marginal prior of $\bm{c}_0$ induced by the mixture model, which depends only on hyperparameters. Switching the order of the summation and integration in \eqref{eq:margleap} and normalizing, we can write the marginal LEAP as
\begin{equation}
    \pi(\bm{\theta}_1 | D_0) = \sum_{\bm{c}_0 \in \{ 1, \ldots, K \}^{n_0}} \pi(\bm{\theta}_1 | \bm{c}_0) \pi(\bm{c}_0 | D_0).
    \label{eq:leap_mixture_prior}
\end{equation}
The equation in \eqref{eq:leap_mixture_prior} indicates that the LEAP is a mixture prior consisting of $K^{n_0}$ components, with the mixture weights being the marginal prior of $\bm{c}_0$ induced by the mixture model. In Section \ref{sec:comparison}, we show that this characterization of the LEAP is closely related to robust mixture priors and Bayesian model averaging (BMA).

The joint posterior density under the LEAP is given by 
\begin{align}
    p(\bm{\theta}, \bm{\gamma}, \bm{c}_0 &| D, D_0)
       \propto \pi_0(\bm{\theta}, \bm{\gamma}) \left[ \left( \prod_{i=1}^{n} f(D_{i} | \bm{\theta}_1) \right) \left( \prod_{\{ c_{0i} = 1 \}} f(D_{0i} | \bm{\theta}_1) \right) \right] \gamma_1^{n_{01}}
       \notag \\
       &
       \times 
       \prod_{k=2}^K \left[ \gamma_k^{n_{0k}}
       \prod_{i=1}^{n_0} f(D_{0i} | \bm{\theta}_k)^{c_{0ik}}
       \right],
       \label{eq:leap_jointpost}
\end{align}
where, $\{ c_{0i} = 1 \} = \{ i \in \{1, \ldots, n_0\} : c_{0i} = 1 \}$ contains the indices for the historical data classified to be exchangeable and $n_{0k} = \sum_{i=1}^{n_0} {c_{0ik}}$ is the number of subjects in class $k \in \{1, \ldots, K\}$. We provide conditions under which the posterior in \eqref{eq:leap_jointpost} is proper in Section \ref{sec:prior_elicitation}. Note that the posterior in \eqref{eq:leap_jointpost}, for a particular value of $\bm{c}_0$, essentially pools the subjects from the current data set and those from the historical data that are classified to be exchangeable. Hence, $n_{01}$ has a nice interpretation, namely, the contribution of the historical data set to the posterior of $\bm{\theta}_1$. We refer to the number of observations contributing to $\bm{\theta}_1$ as ``sample size contribution'' (SSC), which is formally introduced in Section \ref{sec:ssc} in the context of limiting the informativeness of the LEAP.

\subsection{Prior elicitation}
\label{sec:prior_elicitation}
In this section, we discuss prior elicitation for the initial prior for the LEAP. We provide conditions for which the LEAP and the posterior density under the LEAP are proprer.

Recall that the LEAP is essentially the prior induced by the posterior of a finite mixture model for the historical data set. Thus, in the case of unbounded parameters, a proper initial prior $\pi_0(\bm{\theta})$ is required for the LEAP to be proper. This is most easily observed using the data augmentation scheme in \eqref{eq:jointleap}. If a component of the finite mixture is unoccupied (i.e., $n_{0k} = 0$ for some $k$), then the LEAP for the $k^{th}$ component is the marginal prior $\pi_0(\bm{\theta}_k)$. 

Conversely, the posterior density in \eqref{eq:leap_jointpost} may be proper even if the initial prior $\pi_0(\bm{\theta})$ is improper. We provide conditions for posterior propriety in Theorem \ref{thm:postproper}.
\begin{theorem}
    Let $L(\bm{\theta}_1 | D)$ denote the likelihood of the current data set and let prior be the LEAP given by \eqref{eq:leap_mixture_prior}. If, for every partition $\{ 1, \ldots, K \}^{n_0}$, the following conditions hold, the posterior density $p(\bm{\theta} | D, D_0)$ in \eqref{eq:leap_jointpost} is proper
    \begin{enumerate}
        \item $\int 
        \pi_0(\bm{\theta}_1)
        \left[ \prod_{i=1}^n f(D_i | \bm{\theta}_1) \right]
        \left[ \prod_{\{ i : c_{0i1} = 1 \}}^{n_0} f(D_{0i} | \bm{\theta}_1) \right]
        d\bm{\theta}_1 < \infty$, and
        \item $\int \int \gamma_1^{n_{01}}
        \prod_{k=2}^K \gamma_k^{n_{0k}} \prod_{i=1}^{n_0} f(D_{0i} | \bm{\theta}_k)^{c_{0ik}}
        \pi_0( \bm{\theta}_{(-1)}, \bm{\gamma} | \bm{\theta}_1)
        d\bm{\theta}_{(-1)}d\bm{\gamma} < \infty$ for every $\bm{\theta}_1$.
    \end{enumerate}
    \label{thm:postproper}
\end{theorem}
The first condition of Theorem \ref{thm:postproper} says that each partition of the historical data set must yield a proper posterior based on the marginal initial prior for $\bm{\theta}_1$. In general, the worst case scenario would be when all subjects in the historical data set are assigned to classes other than the first class. A simplified version of the first condition in Theorem \ref{thm:postproper} is that the posterior density is proper under initial prior $\pi_0(\bm{\theta}_1)$ without using the historical data.

The second condition of Theorem \ref{thm:postproper} essentially says that the initial prior must be proper for all components other than the first, which is analogous to the conditions for propriety in finite mixture models. This is intuitive since the likelihood of the current data does not involve the parameters of the other components. For the remainder of this paper, we assume initial priors of the form 
$
    \pi_0(\bm{\theta}, \bm{\gamma}) = \pi_0(\bm{\gamma}) \prod_{k=1}^K \pi_0(\bm{\theta}_k).
    %
$

In general, it is desirable to elicit a prior on $\bm{\gamma}$ that is noninformative over the $K$-dimensional simplex so that the data can decide how best to allocate the components of the mixture. In Section \ref{sec:asymptotic_leap}, we argue that noninformative Dirichlet priors with small concentration parameters play a fundamental role in the posterior consistency of $\bm{\gamma}$.

\subsection{Sample size contribution}
\label{sec:ssc}
We now develop a stochastic definition of SSC that is intuitive and simple. Under a conditionally conjugate (i.e., Dirichlet) prior for $\bm{\gamma}$, the prior SSC is available in closed form. Even if a conjugate prior is not used, the prior and posterior SSC is easily computable using numerical and/or Monte Carlo integration.
It is clear from the joint posterior in \eqref{eq:leap_jointpost} that, for a fixed $\bm{c}_0$, the contribution from the historical data to the posterior of $\bm{\theta}_1$ is given by $n_{01} = \sum_{i=1}^{n_0} c_{0i1}$. Thus, $n_{01}$ can be thought of as a ``sample size contribution'' (SSC). In the sequel, we show that this stochastic concept of the SSC is useful for prior elicitation.

Let $\bm{n}_0 = (n_{01}, \ldots, n_{0K})$ denote the $K$-dimensional vector giving the size of each class for a fixed value of $\bm{c}_0$. Since $\bm{c}_{0i} | \bm{\gamma} \sim \text{Categorical}(\bm{\gamma})$ are i.i.d., we have that $\bm{n}_0 | \bm{\gamma} \sim \text{Multinomial}(n_0, \bm{\gamma})$ and can write the PMF as
$
  f( \bm{n}_0 | \bm{\gamma}) = \frac{n_0!}{n_{01}! \cdots n_{0K}!} \prod_{k=1}^K \gamma_k^{n_{0k}}
$
The conditional distribution of $n_{01} | \bm{\gamma} $ is thus binomial, i.e.,
$
  f(n_{01} | \bm{\gamma}) = \binom{n_0}{n_{01}} \gamma_1^{n_{01}} \left( 1 - \gamma_1 \right)^{n_0 - n_{01}}.
$
Let $\pi_0(\gamma_1) = \int \cdots \int \pi_0(\bm{\gamma}) d\gamma_2 \cdots d\gamma_K$ denote the marginal prior for $\gamma_1$. Then the marginal prior for the SSC is given by
\begin{equation}
   \pi_0(n_{01}) \propto \gamma_1^{n_{01}}(1 - \gamma_1)^{n_0 - n_{01}} \pi_0(\gamma_1).
    \label{eq:effective_size_prior}
\end{equation}

As a special case, if $\pi_0(\gamma_1)$ is a beta prior with shape parameters $\delta_{01}, \delta_{02}$, such as when $\pi_0(\bm{\gamma})$ is a Dirichlet prior (in which case $\delta_{01} = \alpha_1$ and $\delta_{02} = \sum_{k=2}^K \alpha_{0k}$), then we have
$
  \pi_0(n_{01}) = \binom{n_0}{n_{01}} \frac{ B(n_{01} + \delta_{01}, n - n_{01} + \delta_{02}) }{B(\delta_{01}, \delta_{02})},
  %
$
i.e., $n_{01}$ is a beta-binomial random variable.

This characterization of the prior SSC facilitates prior elicitation. For example, suppose $K = 2$ and the analyst, after studying the historical data and considering the inclusion-exclusion criteria of the current study, believes with $95\%$ probability that the true number of exchangeable individuals in the historical data set is between $n_{01}^{\text{low}}$ and $n_{01}^{\text{high}}$. 
Then, under a beta prior for $\gamma_1$, the analyst may use the quantile function of the beta-binomial distribution to find values of hyperparameters $\alpha$ and $\delta$ that reflect this belief. Under a non-conjugate prior for $\gamma_1$, then the PMF in \eqref{eq:effective_size_prior} is easily obtainable using numerical integration and it is straightforward to find hyperparameters that satisfy such criteria.


\subsection{Larger historical data sets}
Occasionally, the historical data sample size $n_0$ exceeds the current data sample size $n$. For example, the ESTEEM I trial is much larger than the ESTEEM II trial. It would be inappropriate to utilize the LEAP as currently developed since it is feasible for the historical data to dominate the posterior. To combat this, we derive a new density over the $K$-dimensional simplex where one of the components is truncated.


Consider the density
\begin{align}
  f(\bm{\gamma} | \bm{\alpha}, a, b) &= \frac{1}{C(\bm{\alpha}, a, b)} 1\{ a < \gamma_1 < b \} \prod_{k=1}^K \gamma_k^{\alpha_k - 1}, \bm{\gamma} \in \Omega,
  \label{eq:ptdd}
\end{align}
where $0 \le a < b \le 1$, $\Omega = \{ \bm{\gamma} \in [0,1]^K : \sum_{k=1}^K \gamma_k = 1, a < \gamma_1 < b \}$, and $C(\bm{\alpha}, a, b) = \int_{\Omega} 1\{ a < \gamma_1 < b \} \prod_{k=1}^K \gamma_k^{\alpha_k - 1} d\bm{\gamma}$ is a normalizing constant. We call the distribution corresponding to the density in \eqref{eq:ptdd} a \emph{partially truncated Dirichlet} (PTD) distribution with concentration parameter $\bm{\alpha}$ and truncation parameters $a$ and $b$, and we write
$
  \bm{\gamma} \sim \text{PTD}\left( \bm{\alpha}, a, b \right).
$

Let $\bm{\gamma} \sim \text{PTD}(\bm{\alpha}, a, b)$ and decompose $\bm{\gamma} = (\gamma_1, \bm{\gamma}_2)'$. Let $\alpha_0 = \sum_{k=2}^K \alpha_k$. The following properties of the PTD distribution hold:
\begin{enumerate}
    \item $\gamma_1 \sim \text{Truncated-Beta}(\alpha_1, \alpha_0, a, b)$.
    \item $\bm{\gamma}_2 | \gamma_1 \sim (1 - \gamma_1) \times \text{Dirichlet}(\bm{\alpha}_2)$.
    \item If $\bm{x} | \bm{\gamma} \sim \text{Multinomial}(\bm{n}, \bm{\gamma})$ where $\bm{n} = (n_1, \ldots, n_K)$ is known, then 
    $
      \bm{\gamma} | \bm{x} \sim \text{PTD}\left( \bm{n} + \bm{\alpha}, a, b \right).
    $
\end{enumerate}
We formally prove statements (1)-(3) in Section 3 of the Supplementary Appendix. Since the marginal prior on $\gamma_1$ is a truncated beta prior, we may elicit $a = 0$ and $b = \min\{n/n_0, 1\}$ to guarantee that the historical data do not dominate the posterior. Using (1), (2), and (3), it is straightforward to sample from the PTD density in a Gibbs sampling algorithm.

\subsection{Asymptotic properties of the LEAP}
\label{sec:asymptotic_leap}
In this section, we discuss the LEAP's asymptotic properties. The asymptotic distribution is crucial to understanding the borrowing properties of the LEAP. We argue that the LEAP can asymptotically pool the current and historical data sets under full exchangeability while being robust to prior-data conflict.

Suppose that the true number of components in the historical data is $K_0 \ge 1$, but we elicit $K > K_0$. Suppose without loss of generality that we may decompose $\bm{\gamma} = (\gamma_1, \ldots, \gamma_{K_0}, \gamma_{K_0 + 1}, \ldots, \gamma_{K})'$. Under some mild regularity conditions, (e.g., the prior for $\bm{\gamma}$ being noninformative relative to the dimension of the $\bm{\theta}_k's$), 
\cite{rousseau_asymptotic_2011} show that 
$
  \sum_{k=K_0 + 1}^{K} E\left[ \gamma_k | D_0 \right] \to 0 \text{ as } n_0 \to \infty.
$
Hence, extra components will, asymptotically, be given no weight in the resulting posterior density of the historical data. For mixtures of location-scale families, \cite{rousseau_asymptotic_2011} recommend eliciting $\bm{\gamma} \sim \text{Dirichlet}(\bm{\alpha})$ with $\alpha_j < 1$ for every $j \in \{1, \ldots, K\}$.

For example, suppose that the current data set and historical data sets are completely exchangeable. Then we have $\gamma_1 = 1$ and $\gamma_2 = \ldots = \gamma_K = 0$. It follows that
$
  \pi(\bm{\theta} | D_0) = \int \pi(\bm{\theta}, \bm{\gamma} | D_0) d\bm{\gamma} \propto L(\bm{\theta}_1 | D_0) \pi_0(\bm{\theta}_1) \pi_0(\bm{\theta}_{(-1)}) \text{ as } n_0 \to \infty,
$
Thus, if the prior for the mixture weights is noninformative, the LEAP asymptotically pools the current and historical data.

Conversely, if no one in the historical data set is exchangeable with the current data set, then $\gamma_1 = 0$. By an analogous argument, under this case, the marginal prior for $\bm{\theta}_1$ under the LEAP will converge to the initial prior for $\bm{\theta}_1$. That is,
$
  \pi(\bm{\theta}_1 | D_0) \approx \pi_0(\bm{\theta}_1) \text{ as } n_0 \to \infty.
$

The facts that the data are pooled under full exchangeability and the historical data are given no weight under a lack of exchangeability illustrate the potential versatility and robustness of the LEAP. Namely, type I error rates should be relatively low when using the LEAP, while efficiency gains can be made even if only a subset of the historical data is exchangeable. This conjecture is verified via simulation studies presented in Section \ref{sec:sims}.

\section{Comparison to other methods}
\label{sec:comparison}
In this section, we compare the LEAP with Bayesian model averaging, the power prior, and PS approaches. We discuss situations under which the LEAP can be competitive in comparison to the other priors.
\vspace{-0.4cm}



\subsection{Relationship with Bayesian model averaging}
\label{sec:bma}
Bayesian model averaging (BMA) is considered a gold standard in model uncertainty. In this application, our uncertainty is regarding who from the historical data set is exchangeable with the current data set. It is an ideal, then, that we would set a prior probability for each partition, pool the current data and the historical data assigned to the first class, and average over posterior results. In this section, we show that the LEAP induces a prior probability on the space of partitions, and develop the notion of a posterior partition probability.

Note that we may write the joint LEAP as $\pi(\bm{\theta}) = \sum_{\{1, \ldots, K\}^{n_0}} \pi(\bm{\theta}| \bm{c}_0, D_0) \pi(\bm{c}_0 | D_0)$, where
\begin{align}
  \pi(\bm{\theta} | \bm{c}_0, D_0) =\frac{ 
    \int \pi_0(\bm{\theta}, \bm{\gamma})\prod_{k=1}^K \gamma_k^{n_{0k}} \prod_{i=1}^{n_0} f(D_{0i} | \bm{\theta}_k)^{c_{0ik}}
    d\bm{\gamma}
  }
  {
    \int \int \pi_0(\bm{\theta}, \bm{\gamma})  \prod_{k=1}^K \gamma_k^{n_{0k}} \prod_{i=1}^{n_0} f(D_{0i} | \bm{\theta}_k)^{c_{0ik}}d\bm{\gamma} d\bm{\theta}
  }
  \label{eq:prior_theta_given_c0}
\end{align}
is the conditional prior for $\bm{\theta}$ based on a particular classification $\bm{c}_0$ and
\begin{align}
\pi(\bm{c}_0 | D_0) = 
  \frac
  {\int \int \pi_0(\bm{\theta}, \bm{\gamma}) \prod_{k=1}^K \prod_{i=1}^{n_0}f(D_{0i} | \bm{\theta}_k)^{c_{0ik}} d\bm{\theta} d\bm{\gamma} }  
  {\sum_{\{1, \ldots, K\}^{n_0}} \int \int \pi_0(\bm{\theta}, \bm{\gamma}) \prod_{k=1}^K \prod_{i=1}^{n_0}f(D_{0i} | \bm{\theta}_k)^{c_{0ik}} d\bm{\theta} d\bm{\gamma} }
  \label{eq:marg_prior_c0}
\end{align}
is the marginal PMF of $\bm{c}_0$ induced by the historical data and initial priors. The denominator in \eqref{eq:marg_prior_c0} may be interpreted as a marginal likelihood for the mixture model of the historical data, and the numerator may be interpreted as a ``partitional marginal likelihood,'' i.e., the marginal likelihood based on a particular partition of $D_0$.

Let the ``partition space'', defined as set of all possible partitions of $\{ 1, \ldots, K \}^{n_0}$, be denoted by $\mathcal{D}_0 = \{ \mathcal{D}_{0l}, l = 1, \ldots, L \}$ where $\mathcal{D}_{0l} = \{ D_{0l1}, \ldots, D_{0lK} \}$ is the $l^{th}$ data partition, $D_{0lk} = \{ i \in \{1, \ldots, n_0\} : c_{0ik} = 1 \}$ is the set of indices in $\mathcal{D}_{0l}$ belonging to class $k$, and $L = K^{n_0}$. 
Let $\pi(\mathcal{D}_{0l})$ denote the probability from the PMF $\pi(\bm{c}_0 | D_0)$ corresponding with the $l^{th}$ partition and let $\pi(\bm{\theta} | \mathcal{D}_{0l})$ be the prior in \eqref{eq:prior_theta_given_c0} corresponding to partition $\mathcal{D}_{0l}$. We refer to the $\pi(\mathcal{D}_{0l})$'s as ``prior partition probabilities.'' The prior for $\bm{\theta}$ may then be expressed as
\begin{align}
  \pi(\bm{\theta} | D_0) \propto \sum_{l=1}^{K^{n_0}}  \left[ \prod_{k=1}^K L(\bm{\theta}_k | D_{0lk}) \right] \pi_0(\bm{\theta}) \pi(\mathcal{D}_{0l}) 
  ,
  \label{eq:leap_bma}
\end{align}
where $L(\bm{\theta}_k | D_{0lk})$ is the likelihood function of $\bm{\theta}_k$ based on data $D_{0lk}$. Note that the prior in \eqref{eq:leap_bma} is equivalent to fitting the posterior for all $K^{n_0}$ data partitions and averaging over the prior partition probabilities, which is conceptually equivalent to BMA. 

Using the LEAP in \eqref{eq:leap_bma}, the posterior density for $\bm{\theta}_1$ is given by \\
$
    p(\bm{\theta} | D, D_0) 
    = \sum_{l=1}^{K^{n_0}} p(\bm{\theta} | D, \mathcal{D}_{0l}) p(\mathcal{D}_{0l} | D, D_0),
$
where 
$
    p(\bm{\theta} | D, D_{0l}) 
    = \frac{ 
        \pi_0(\bm{\theta}) L(\bm{\theta}_1 | D, D_{0l1}) \prod_{k=2}^K L(\bm{\theta}_k | D_{0lk})
    }{
    \int \pi_0(\bm{\theta}) L(\bm{\theta}_1 | D, D_{0l1}) \prod_{k=2}^K L(\bm{\theta}_k | D_{0lk}) d\bm{\theta}
    }
$
is the posterior probability of $\bm{\theta}$ based on historical data partition $\mathcal{D}_{0l}$ and where $p(\mathcal{D}_{0l} | D, D_0) = \int p(\bm{\theta} | D, \mathcal{D}_{0l}) d\bm{\theta}$ is the posterior partition probability (marginal likelihood) of the $l^{th}$ historical data partition, which mimics the posterior model probability in BMA. Hence, the posterior distribution of the LEAP may be interpreted as a ``partition averaged'' posterior density, i.e., the posterior distribution is a weighted average of partitions of the historical data set. The partitions are weighted based on its prior and how well the partition fits the current data (i.e., the ``posterior partition probabilities''). 

For example, consider an i.i.d. Poisson model with $K = 2$ components for the LEAP. Let the initial prior for $\bm{\theta} = (\theta_1, \theta_2)$ be a product of Gamma$(\eta_{0k}, \beta_{0k})$ densities, $k = 1, 2$ and that for $\bm{\gamma} =(\gamma_1, \gamma_2)$ be a Dirichlet$(\bm{\alpha}_0)$ density, where $\bm{\alpha}_0 = (\alpha_{01}, \alpha_{02})'$. Let $n_{0k} = \sum_{i=1}^{n_0} c_{0ik}$ denote the number of subjects assigned to class $k$ and let $\bar{y}_{0k}$ denote the sample mean for the particular partition, where we take $\bar{y}_{0k} = 0$ if $n_{0k} = 0$. It can be shown that the posterior partition PMF is given by
$
  p(\bm{c}_0 | D, D_0) \propto B(\bm{n}_0 + \bm{\alpha}_0) \prod_{k=1}^K \Gamma(\eta_k) \beta_k^{-\eta_k},
$
where $B(\cdot)$ is the multivariate beta function, $\eta_1 = n\bar{y} + n_{01} \bar{y}_{01} + \eta_{01}$, $\beta_1 = n + n_{01} + \beta_{01}$, $\eta_2 = \eta_{02}$ and $\beta_2 = \beta_{02}$. Table~\ref{tab:partitionprobs} shows the prior and posterior partition probabilities and means for each partition assuming $\bar{y} = 1.5$, $n = 10$, $\bm{y}_0 = (1, 2, 6)$, $\bm{\alpha}_0 = (0.9, 0.9)$, and $\theta_k \sim \text{Gamma}(0.1, 0.1)$. 
\begin{table}
    \centering
    \begin{tabular}{lcccc}
        \toprule
        $(c_{01}, c_{02}, c_{03})$  & $E[\theta_1 | \bm{y}_0, \bm{c}_0]$ & $\pi(\bm{c}_0 | \bm{y}_0)$ & $E[\theta_1 | \bm{y}, \bm{y}_0, \bm{c}_0]$ & \multicolumn{1}{c}{$p(\bm{c}_0 | \bm{y}, \bm{y}_0)$} \\ 
        \midrule
        $(1, 1, 1)$  & $2.94$ & $0.319$ & $1.84$ & $0.412$ \\
        $(2, 2, 2)$  & $1.00$ & $0.319$ & $1.50$ & $0.108$ \\
        $(1, 1, 2)$  & $1.48$ & $0.092$ & $1.50$ & $0.259$ \\
        $(2, 2, 1)$  & $5.55$ & $0.092$ & $1.90$ & $0.017$ \\
        $(1, 2, 1)$  & $3.38$ & $0.020$ & $1.83$ & $0.019$ \\
        $(2, 1, 2)$  & $1.91$ & $0.020$ & $1.54$ & $0.045$ \\
        $(1, 2, 2)$  & $1.00$ & $0.068$ & $1.45$ & $0.105$ \\
        $(2, 1, 1)$  & $3.86$ & $0.068$ & $1.91$ & $0.035$ \\
        \bottomrule 
    \end{tabular}
    \caption{Prior and posterior means and partition probabilities for a Poisson model with $K = 2$, $n = 10$, $\bar{y} = 1.5$, $\bm{y}_0 = (1, 2, 6)$, $\gamma_1 \sim \text{Beta}(0.9, 0.9)$, and $\theta_k \sim \text{Gamma}(0.1, 0.1)$.}
    \label{tab:partitionprobs}
\end{table}

Recall that the prior PMFs are induced by the mixture model of the historical data set. Since $y_{01}$ is closer to $y_{02}$ than $y_{03}$, for a fixed value of $\bm{n}_0$, the prior PMFs that cluster the first and second observations together are given higher weight than those that cluster the first and third or second and third observations together. Thus, the LEAP has the attractive feature that, for a fixed value of $\bm{n}_0$, historical data observations that fit locally to a parametric distribution are given higher prior weight. Table \ref{tab:partitionprobs} also shows that the posterior partition probabilities are influenced by the observe data via a Bayesian update. For example, since $y_{01}, y_{02}$ are closer to the mean of the observed data, higher posterior weight is given for $(1, 1, 2)$ than for $(2, 1, 1)$ even though the prior partition probabilities are equal, illustrating the LEAP's dynamic borrowing properties.

The overall posterior mean is given as $E(\theta_1 | \bm{y}, \bm{y}_0) = \sum_{l=1}^8 p(\bm{c}_{0l} | \bm{y}, \bm{y}_0) E(\theta_1 | \bm{y}, \bm{y}_0, \bm{c}_{0l}) \approx 1.66$, which is identical to the estimated posterior mean using MCMC. An attractive feature of the MCMC approach is that it is computationally feasible in higher dimensions and does not spend time sampling from partitions with low posterior probability.

In Section 4 of the Supplementary Appendix, we provide analytical results of the LEAP for the normal linear model. We show that prior partition probabilities are higher for partitions yielding lower mean squared errors, agreeing with the Poisson example. Furthermore, we show that, all else equal, the posterior model probabilities decrease in $\lVert \hat{\bm{\beta}}_1 - \hat{\bm{\beta}}_{01} \rVert$, where $\hat{\bm{\beta}}_1$ is the MLE of the current data set and $\hat{\bm{\beta}}_{01}$ is the MLE for a particular partition of the historical data set assigned to the first component. Thus, posterior partition probabilities will be larger for (and, hence, more posterior weight will be given to) partitions of the historical data set whose MLE is similar to that of the current data set.

We note that the partition averaging representation is closely related to robust mixture priors, which take the form $\pi(\bm{\theta}_1) = p_0 \pi_I(\bm{\theta}_1) + (1 - p_0) \pi_{V}(\bm{\theta}_1)$, where $p_0 \in (0, 1)$, $\pi_I(\cdot)$ is an informative prior density, and $\pi_V(\cdot)$ is a vague (i.e., noninformative) prior density. For the LEAP, partitions yielding small values of $n_{01}$ will be noninformative priors (e.g., if $n_{01} = 0$, the prior for $\bm{\theta}_1$ is simply the initial prior $\pi_0(\bm{\theta}_1)$).

\subsection{Comparison to power priors}
\label{sec:powerprior}
The power prior (PP) is a popular prior for historical data sets. The PP is simple, and its interpretation as a discounted likelihood is attractive. Given a historical data set $D_0$, the PP is given by
$
    \pi_{\text{PP}}(\bm{\theta}_1 | D_0, a_0) = \frac{L(\bm{\theta}_1 | D_0)^{a_0} \pi_0(\bm{\theta}_1)}{C(a_0)},
$
where $\pi_0(\bm{\theta})$ is an initial prior for the parameters (which is usually taken to be noninformative and, in some cases, can be improper), $a_0 \in [0,1]$ is a discounting parameter, and $C(a_0) = \int L(\bm{\theta}_1 | D_0)^{a_0} \pi_0(\bm{\theta}_1) d\bm{\theta}_1$ is a normalizing constant. 

Note that if $\gamma_1 = a_0 \in \{0, 1\}$ is fixed, then the LEAP and the PP are equivalent. A fundamental difference between the PP and the LEAP is that the PP conducts blanket discounting on \emph{all} individuals in the historical data set. By contrast, the LEAP searches for the most relevant individuals in the historical data set pertaining to the outcome model, so that the most relevant individuals have a higher weight on the posterior density. Thus, in settings where the exchangeability assumption holds, the PP may be more appropriate. If exchangeability does not hold, then the PP can lead to large bias and type I error rates.

While the LEAP is a dynamic prior, the PP is not when $a_0$ is taken as fixed. To solve this issue, the normalized power prior (NPP) has been proposed \citep{duan_evaluating_2006}. The NPP extends the PP by putting a prior on $a_0$. Specifically, the NPP is given by 
$
  \pi_{\text{NPP}}(\bm{\theta}_1, a_0 | D_0) = \pi_{\text{PP}}(\bm{\theta}_1 | D_0, a_0) \pi(a_0),
  \label{eq:npp}
$
where $\pi(a_0)$ is a prior on $a_0 \in [0,1]$. The NPP avoids having to specify a single value for $a_0$ and allows the value of $a_0$ to be data driven. \cite{hobbs_commensurate_2012} point out that the posterior of $a_0$ under the NPP closely resembles the prior. Thus, even under full exchangeability and a large sample, the posterior mean of $a_0$ will not approach $1$ under a flat prior for $a_0$. However, the posterior mean of $a_0$ will typically approach 0 for flat priors on $a_0$ under incompatibility between the current and historical data sets. If a proportion of individuals in the historical data set are exchangeable, it is unclear how the NPP will behave, but these limiting cases shed some light. If the posterior of $a_0$ is close to zero, then bias (in terms of the posterior expectation of the parameters) will be minimized but efficiency gains are limited. 
Conversely, if the posterior of $a_0$ resembles the prior, then point estimates under the NPP may have substantial bias.

By contrast, we argued in Section \ref{sec:asymptotic_leap} that when the samples size for the historical data is large enough, the LEAP has the effect of pooling the historical data and current data sets together when the data sets are completely exchangeable, while being completely uninformative asymptotically if there is complete incompatibility between the two data sets. We numerically examine this property of the LEAP via simulation in Section \ref{sec:sims}.

In general, there is a computational and implementation advantage for the LEAP over the NPP. Namely, the NPP requires the evaluation of $C(a_0)$ in every iteration of the MCMC scheme in general, and $C(a_0)$ is only analytically tractable under certain special cases (e.g., a normal linear model with a conjugate prior). By contrast, the LEAP simply specifies a mixture model for the historical data set, allowing for non-conjugate priors and avoiding having to compute a normalizing constant under those contexts. In general, it is quite easy to implement the LEAP for any MCMC software that supports mixture models.

\subsection{Comparison to propensity score approaches}
\label{sec:ps_integrated}
More recently, so-called ``propensity score integrated'' priors have been developed, which expand upon existing priors under the assumption that balance in baseline characteristics yields exchangeability. Specifically, let $D_0 = \{ (y_{0i}, z_{0i}, \bm{x}_{0i}), i = 1, \ldots, n_0 \}$ denote the historical data and let $D = \{ (y_i, z_i, \bm{x}_i), i = 1, \ldots, n \}$ denote the current data, where, for the historical and current data sets, respectively, $y_{0i}$ and $y_i$ are responses, 
$z_{0i}$ and $z_i$ are treatment indicators (equaling 1 if the individual was treated; 0 if the individual received control), 
and $\bm{x}_{0i}$ and $\bm{x}_i$ are vectors of covariates, each of which may include an intercept term.

PS integrated approaches are two-step approaches. The first step, which we refer to as the ``design stage,'' involves estimating the probability of being in the current study as a function of covariates (i.e., the PS). Note that this PS is different than that in the causal inference literature, which models the probability of receiving treatment as a function of covariates. Let $s_i = 1$ for $i = 1, \ldots, n$ and let $s_{0i} = 0$ for all $i = 1, \ldots, n_0$ denote the indicator to which study each individual belongs. Let $e_i = \widehat{\Pr}(s_i = 1 | \bm{x}_{i})$ and $e_{0i} = \widehat{\Pr}(s_{0i} = 1 | \bm{x}_{0i})$ denote the estimated PSs for the current and historical data sets, respectively, that are obtained from a logistic regression model after pooling the two data sets. Then, PS approaches (e.g., matching, weighting, stratification) are conducted in attempt to balance the baseline covariates. 

The second step, which we call the ``analysis stage,'' is to analyze the data obtained from step 1 (excluding the covariates). The underlying assumption is that the PS approach in the design stage balances baseline covariates, so that there is nothing further to control for.

For example, \cite{lu_propensity_2022} propose the PS integrated power prior (PSIPP), where observations are stratified by the PS in the design stage and strata-specific PPs are elicited in the analysis stage. Their prior may be represented as
$
    \pi_{\text{PSIPP}}( \bm{\mu} | \bm{a}_0) = \prod_{z=0}^1 \prod_{j=1}^J L(\mu_{zj} | D_{0zj})^{a_{0j}} \\ \times \pi_0(\mu_{zj}),
    \label{eq:psipp}
$
where $\mu_{zj}$ is the mean for individuals receiving treatment $z$ in the $j^{th}$ stratum, $D_{0zj}$ contains the historical data assigned to treatment $z$ and the $j^{th}$ stratum, $a_{0j}$ is a stratum-specific discounting parameter, and $\pi_0(\cdot)$ is an initial prior. A treatment effect, $\Delta$, is then obtained via an average of the stratum-specific treatment effects, i.e.,
$
  \Delta = \frac{1}{J} \sum_{j=1}^J (\mu_{1j} - \mu_{0j}).
$

There are several disadvantages with PS integrated approaches. First, these approaches are not fully Bayesian. Note that the PS is estimated, but it does not take into account uncertainty surrounding the parameter values of the PS. As a result, each individual is placed into a single stratum during the analysis stage, and uncertainty regarding exchangeability is not incorporated in these approaches. By contrast, in Section \ref{sec:bma}, we showed that the LEAP can be conceptualized as a mixture prior over partitions of exchangeable individuals, directly incorporating this uncertainty.

While it may seem desirable to incorporate a fully Bayesian approach, allowing the PS and outcome parameters to be jointly estimated, it has been shown that these approaches generally lead to increased bias in the Bayesian causal inference literature \citep{zigler_model_2013}. A quasi-Bayesian approach that cuts the dependence from the PS to the outcome model could alternatively be used, but that approach is also not fully Bayesian.

Perhaps most critically, these approaches implicitly assume that individuals are exchangeable if and only if their covariates are similar. In most regression settings, we condition on covariates, e.g., we model $E(Y | \bm{X} = \bm{x})$. When adjusting for covariates, we do not require covariate balance for inference to be valid. If there is a discrepancy between the covariate distributions of the historical and current data set but the distribution of the outcome conditional on the covariates is the same, PS approaches can be inefficient since exchangeable individuals may be discarded.

\section{Simulations}
\label{sec:sims}
In this section, we present extensive simulation results comparing the LEAP with several competing priors, including a normalized version of the partial borrowing power prior \citep{ibrahim_bayesian_2015}, which we refer to as the NPBPP, the PS integrated power prior (PSIPP), and a reference prior. The simulations are based off of the ESTEEM I and ESTEEM II trials described in Section \ref{sec:motivation}. To compare the borrowing properties of the priors, we generate current data sets that are at least as large as the historical data sets. 

\subsection{Simulation setup}
For the ESTEEM trials, the primary outcome measure was the percentage of participants who achieved an improvement of at least 75 percent in the Psoriasis Area Severity Index (PASI) at week 16 from baseline. A PASI score can range from 0 to 72 with higher scores indicating more severe psoriasis. We consider the percent reduction in the PASI score as the outcome of interest as opposed to the dichotomization of it.

Let $D = \{ (y_i, z_i, \bm{x}_i), i = 1, \ldots, n \}$ denote the current data, where $z_i \in \{0, 1\}$ is a treatment indicator ($z_i = 1$ if subject $i$ received treatment; $z_i = 0$ if subject $i$ received control) and $\bm{x}_i$ is a $p$-dimensional vector of covariates, which may contain an intercept term. We generate the data from a linear regression model of the form
$
  \bm{y} = \bm{X}\bm{\beta}_{11} + \bm{z} \beta_{12} + \bm{\epsilon},
$
where $\bm{\beta}_{11} = (-18.00, 0.49, -1.67, 1.99)'$ is a vector of regression coefficients associated with the intercept and the centered and scaled covariates (age, age$^2$, and baseline PASI score), $\beta_{12} = -35.39$ is the treatment effect, and $\bm{\epsilon} \sim N_n(0, \tau_1^{-1} \bm{I}_n)$ is an error term. Let the historical control data be denoted by $D_0 = \{ (y_{0i}, \bm{x}_{0i}), i = 1, \ldots, n_{00} \}$, where $y_{0i}$ is the percent change in PASI score for subject $i$, $\bm{x}_{0i}$ is a $p$-dimensional vector of covariates, and $n_{00}$ denotes the number of historical controls. The treatment indicators for the current and historical data sets are generated via $z_i, z_{0j} \sim \text{Bernoulli}(2/3)$ for $i = 1, \ldots, n$, $j = 1, \ldots, n_0$. 
We consider current data sample sizes as $n = n_0 + h$ for $h \in \{ 0, 50, 100, 150 \}$. Note that we omit the treatment indicator since each subject in the historical data receives control. The logarithms of age and baseline PASI were generated from a multivariate normal distribution using the sample mean and covariance matrix from the ESTEEM II study, given by
$
\bm{\mu} = (3.78, 2.90)'
$
with variances and covariance given respectively by $\sigma_1^2 = 0.09$, $\sigma_2^2 = 0.10$, and $\sigma_{12} = -0.01$.

We consider several cases of exchangeability for our simulation, which we refer to fully exchangeable (i.e., the historical data participants have the same outcome and covariate parameters), half exchangeable (i.e., half of the historical data participants have the same outcome and covariate parameters), and fully unexchangeable (i.e., the historical data participants have different outcome and covariate parameters). We assume all individuals in the current data set are exchangeable, and any sources of heterogeneity arise only in the historical data sets. We let $\bm{\beta}$ and $\bm{\mu}$, defined above, denote the regression coefficients and covariate parameters for the exchangeable group.

For the half exchangeable and fully unexchangeable settings, we define a parameter $q \in \{0.25, 0.50, 0.75\}$ that enters as a scalar multiple, namely, $\bm{\beta}_{\text{unexch}} = q \bm{\beta}$ and $\bm{\mu}_{\text{unexch}} = q \bm{\mu}$, where $\bm{\beta}_{\text{unexch}}$ and $\bm{\mu}_{\text{unexch}}$ are the regression coefficients and the mean for the continuous covariates, respectively, for the unexchangeable group. Taking scalar multiples to both the outcome and covariate parameters puts the LEAP and PSIPP on a level playing field. A total of 20,000 posterior samples were obtained after a burn-in period of 2,000 via the Stan programming language \cite{carpenter_stan_2017}.

\subsection{Prior elicitation for the simulation study}
\label{sec:sims_prior_elicit}
The NPBPP is given by
$
    \pi_{\text{NPBPP}}(\bm{\beta}_{11}, \beta_{12}, \tau_1, a_0 | D_0) = \frac{1}{C(a_0)}
    L(\bm{\beta}_{11}, \tau_1 | D_0)^{a_0} \pi_0(\bm{\beta}_{11}, \beta_{12}, \tau_{1}) \pi(a_0),
    %
$
where $C(a_0)$ is a normalizing constant and $\pi_0(\cdot)$ is an initial prior. We assume the initial prior takes the form
$
  \pi_0(\bm{\beta}_{11}, \beta_{12}, \tau_1) \propto \tau_1^{\delta_0 / 2 - 1} \exp\left\{ -\frac{\tau_1}{2} \nu_0 \right\} \exp\left\{ -\frac{\xi_0}{2}(\beta_{12} - \mu_{0})^2 \right\},
$
where $\nu_0$, $\delta_0$, $\xi_0$, and $\mu_{0}$ are elicited hyperparameters. This choice of initial prior allows the normalizing constant $C(a_0)$ for the NPBPP to be analytically tractable \citep{ibrahim_power_2015}, which facilitates MCMC sampling. We elicited $\delta_0 = 0.02$, $\nu_0 = 0.02$, $\xi_0 = 10^{-2}$, and $a_0 \sim U(0, 1)$. 

A second comparator is PSIPP, which is described in detail in Section \ref{sec:ps_integrated}. For our simulations, we stratified by quintiles of the PS. While this selection is somewhat arbitrary, in the causal inference literature, \cite{rosenbaum_reducing_1984} note that stratification based on quintiles can reduce up to 90\% of the bias due to confounding. As suggested by \cite{lu_propensity_2022}, the elicitation of $a_{0j}$ depends on the degree of overlap in the PS distributions of the strata. Stratification was implemented using the \texttt{psrwe} R package \citep{wang_psrwe_2022}.

Finally, we compare the posterior under the LEAP with a reference prior, where we assume $\bm{\beta}_{11} \sim N_p(0, 10^2 \bm{I}_p)$ and $\sigma_1 = \tau_1^{-1/2} \sim N^{+}(0, 10^2)$, where $N^{+}(a, b)$ denotes the positive half-normal distribution with mean $a$ and variance $b$. We may write this prior as
$
    \pi_{\text{ref}}(\bm{\beta}_1, \sigma_1) = 2\phi_p(\bm{\beta}_1 | 0, 10^2) \phi(\sigma_1 | 0, 10^2) 1\{\sigma_1 > 0\}.
    %
$
In order to showcase the broad computational feasibility of the LEAP, we do not use conjugate priors or Gibbs sampling. Instead, we take the initial prior in \eqref{eq:jointleap} as
$
    \pi_0(\bm{\beta}, \bm{\sigma}) = \prod_{k=1}^K \pi_{\text{ref}}(\bm{\beta}_k, \sigma_k),
    %
$
where we set $K = 2$ and $\gamma_1 \sim \text{Beta}(0.95,0.95)$.
\vspace{-0.25cm}

\subsection{Simulation results}
For each simulation scenario, we compute percent absolute bias (PAB), mean squared error (MSE), and 95\% symmetric credible interval (CI) coverage. The PAB is computed as $\text{PAB} = \frac{1}{10,000} \sum_{m=1}^{10,000} \left\lvert   \frac{ \hat{\beta}_{12}^{(m)} - \beta_{12} } { \beta_{12} } \right\rvert$, where $\hat{\beta}_{12}^{(m)}$ is the posterior mean of the treatment effect for the $m^{th}$ data set. The MSE is computed as $\text{MSE} = \frac{1}{10,000} \sum_{m=1}^{10,000} \left( \beta_{12}^{(m)} - \beta_{12} \right)^2 $. The CI coverage is computed as the proportion of the samples with CI intervals containing $\beta_{12}$.

The results of the simulation are presented in Figure~\ref{fig:sims}. The three rows of Figure~\ref{fig:sims} correspond to the fully exchangeable, half exchangeable, and fully unexchangeable settings, respectively. The three columns of the figure correspond to PAB, MSE, and CI coverage respectively, which are measured on the vertical axis. The different priors are represented as colors and the different values of $q$ are represented as line types.

\begin{figure}[ht]
    \begin{center}
        \includegraphics[width=\textwidth]{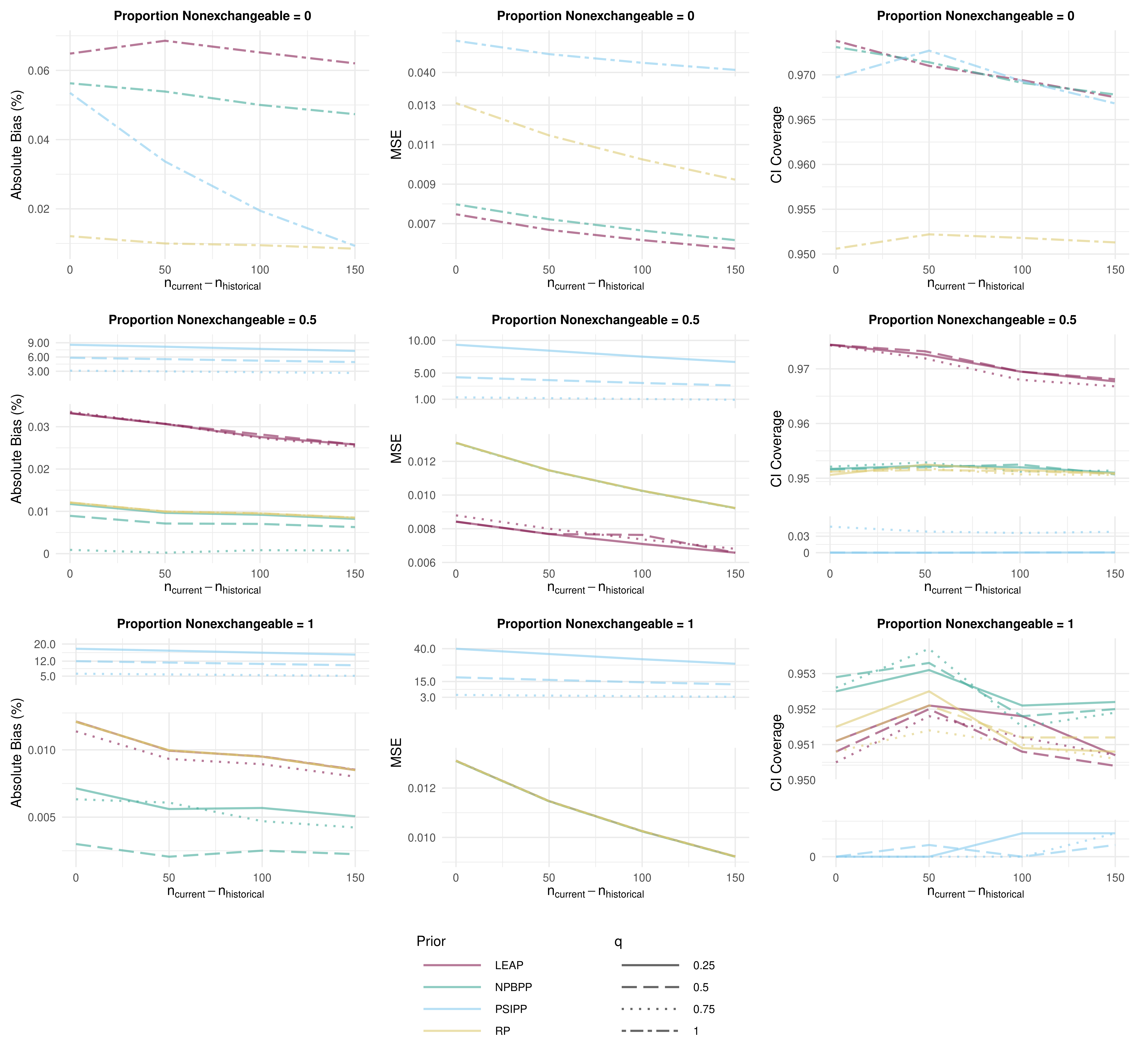}
    \end{center}
    \caption{Percent absolute bias, mean squared error (MSE), and 95\% credible interval (CI) coverage for each prior in the simulation exercise. LEAP = latent exchangeability prior; NPBPP = normalized partial borrowing power prior; PSIPP = propensity score integrated power prior; RP = reference prior}
    \label{fig:sims}
\end{figure}

For the fully exchangeable setting (i.e., the first row), the LEAP has the highest PAB and the lowest MSE, which can be attributed to a bias-variance tradeoff. The bias between the priors, however, is negligible, as the maximum observed PAB for the LEAP is less than 0.08\%. The LEAP and NPBPP have similar MSEs, while the MSE under the PSIPP is over $6$ times that of the LEAP. All three historical data priors had comparable CI coverage.

The PSIPP has poor performance under partial exchangeability and full unexchangeability, where it must be represented as a separate axis. The CI coverage is near 0 for all simulation scenarios and the bias is inflated up to 20\%. This may be because the PSIPP is a static borrowing prior, where the amount of borrowing is taken to be fixed based on the PS (which is treated as known and fixed). 

Under full unexchangeability (i.e., the third row), the LEAP and the NPBPP are again similar in terms of PAB, MSE, and CI coverage, where they essentially have the same features as a reference prior. Where they differ the most is under half exchangeability (i.e., the second row). We see that the LEAP has slightly inflated bias compared to the reference prior and the PBNPP, with an MSE that is about 33\% less than the reference prior across the sample sizes. The LEAP also is the only prior to have higher than nominal CI coverage. This validates the conjecture in Section \ref{sec:leap_general} that efficiency gains can be achieved even if only a fraction of individuals in the historical data are exchangeable.

Overall, it appears that the LEAP is the best performer among these three priors. Having CI coverage at the nominal level under a fully unexchangeable setting is similar to type I error control, which suggests that the LEAP may be a good option when borrowing from external controls. Only the LEAP exhibited efficiency gains under partial exchangeability whereas the performance of the information borrowing priors was similar under full exchangeability.

\section{Data analysis example}
\label{sec:analysis}
In this section, we analyze data from ESTEEM II study, augmenting the control arm using the placebo data from the ESTEEM I study. Key summary statistics for the ESTEEM II trial and the placebo arm in ESTEEM I trial are provided in Table 1 in the Supplementary Appendix. Posterior samples were taken using two variations of the LEAP (i.e., $K = 2$ and $K=3$). We compare posterior summary statistics under the LEAP with the PSIPP, PBNPP, and reference priors described in Section \ref{sec:sims_prior_elicit}.

We constrain all information borrowing priors so that a maximum of $n_{00}^* \le n_{11} - n_{10} = 133$ subjects are borrowed from the historical controls, where $n_{jk}$ is the number
of subjects enrolled in arm $k$ of study $j$, for $k, j \in \{0, 1\}$, where $k = 0$ denotes the control arm and $j = 0$ denotes the historical study. Translated to a proportion, this results in a constraint of $p_{00}^* \le \frac{n_{11} - n_{10}}{n_{00}} \approx 0.48$ for $\gamma_1$ in the LEAP and $a_0$ in the NPBPP. The DIC \citep{spiegelhalter_bayesian_2002}, historical SSC, and the estimated treatment effect on percent change in total PASI Score at Week 16 are evaluated. We report posterior means, standard deviations, and 95\% credible intervals of the treatment effect for each prior.

Using the Stan programming language, 250,000 posterior samples were obtained for each prior after a burn-in period of 2,000. Results for the treatment effect from the data analysis are presented in Table~\ref{tab:results}. Posterior means and standard deviations are provided for all model parameters in Section 5 of the Supplementary Appendix. All priors provided overwhelming evidence that Apremilast results in a lower percent change in PASI scores compared to individuals receiving placebo. 

\begin{table}
    \begin{center}
\begin{tabular}{lccccc}
\toprule
Prior  & DIC & Post. Mean & Post. SD & \multicolumn{2}{c}{$95\%$ CI} \\ 
\midrule
LEAP $(K=2)$  & $4063.27$ & $-31.4$ & $3.35$ & \multicolumn{1}{@{\hspace{\tabcolsep}(}c@{,}}{$-38.0$} & \multicolumn{1}{c@{)\hspace{\tabcolsep}}}{$-24.9$} \\
LEAP $(K=3)$  & $4063.26$ & $-31.5$ & $3.27$ & \multicolumn{1}{@{\hspace{\tabcolsep}(}c@{,}}{$-37.9$} & \multicolumn{1}{c@{)\hspace{\tabcolsep}}}{$-25.0$} \\
PBNPP  & $4062.09$ & $-31.9$ & $3.09$ & \multicolumn{1}{@{\hspace{\tabcolsep}(}c@{,}}{$-37.8$} & \multicolumn{1}{c@{)\hspace{\tabcolsep}}}{$-25.8$} \\
PSIPP  & $4091.21$ & $-28.1$ & $2.69$ & \multicolumn{1}{@{\hspace{\tabcolsep}(}c@{,}}{$-33.3$} & \multicolumn{1}{c@{)\hspace{\tabcolsep}}}{$-22.8$} \\
Reference  & $4064.04$ & $-32.1$ & $3.45$ & \multicolumn{1}{@{\hspace{\tabcolsep}(}c@{,}}{$-38.9$} & \multicolumn{1}{c@{)\hspace{\tabcolsep}}}{$-25.4$} \\
\bottomrule 
\end{tabular}
        \caption{Summary of the posterior density of the treatment effect (mean difference in \% change in PASI score) using the ESTEEM I historical controls and ESTEEM II data sets. DIC = deviance information criterion; Post. Mean = posterior mean; Post. SD = posterior standard deviation, CI = credible interval.}
        \label{tab:results}
    \end{center}
\end{table}

As expected, all information borrowing priors resulted in a lower posterior variance for the treatment effect compared to the reference prior, illustrating the benefit of utilizing the historical data. The posterior variance for the treatment effect is lower for the PBNPP than those for the LEAPs. As the study design (e.g., inclusion/exclusion criteria) was very similar for the ESTEEM I and ESTEEM II studies, it is conceivable that all patients in the historical data set are exchangeable with those in the current data set. As discussed in Section \ref{sec:leap_general}, PPs can be preferable to the LEAP under fully exchangeable settings since they provide discounting on the entire historical data set, while the LEAP averages over partitions subject to the probability of being exchangeable.

The DIC is the lowest for the PBNPP, indicating that the PBNPP is the best fitting prior, while the two LEAPs resulted in slightly higher DIC. The DIC was lower for these priors than the reference prior, indicating a low degree of prior-data conflict. The PSIPP resulted in a substantially larger DIC compared to the reference prior, which may indicate some model misspecification. The fact that the posterior mean under the PSIPP is quite different than those under the other priors provides support of this conjecture. Moreover, there were large differences between the stratum-specific means for the PSIPP (see Section 5 of the Supplementary Appendix). The posterior variance is lowest for the PSIPP, indicating that the PSIPP is the most aggressive prior in terms of information borrowing.

The posterior means and 95\% CIs (in parentheses) for $\gamma_1$ for the LEAP priors were  $0.40$ $(0.14, 0.48)$ and $0.41$ $(0.16, 0.48)$ for $K = 2$ and $K = 3$, respectively, and those for $a_0$ in the PBNPP was $0.34$ and $(0.14, 0.47)$, respectively. Thus, the LEAPs and the PBNPP provide a moderate degree of borrowing from the historical data with moderate uncertianty about how much to be borrowed. Since these values are similar but the posterior under the NPBPP resulted in increased precision, the LEAP may discount more aggressively than the NPBPP.

\section{Conclusion}
\label{sec:conc}
In this paper, we developed a novel prior we refer to as the LEAP. The LEAP is a dynamic borrowing prior that is suitable when only a fraction of individuals in a historical data set can be considered to be exchangeable. It may be particularly attractive under situations in which there is a large amount of historical data, and we only wish to borrow from the most relevant individuals in the historical data sets.

One limitation of the LEAP is that it cannot be used for i.i.d. Bernoulli proportion models. This is because the parameters for the mass function for a mixture of Bernoulli random variables are not identifiable. However, the LEAP could be applied to logistic regression models with at least one continuous covariate (e.g., age). 

The LEAP opens up several avenues for future research. First is the development of the LEAP for other outcome models. For example, a LEAP could be developed for time-to-event data, where censoring precludes using the LEAP as developed in this paper. Second, a LEAP could be developed to borrow information from real-world data (RWD), where confounding is of paramount concern. Although our results indicate that PS approaches should be used with caution, a LEAP that incorporates the PS to formulate individual-specific probabilities of being exchangeable via a hierarchical scheme could be highly relevant since the amount of borrowing would still be outcome-adaptive.

\bibliographystyle{unsrtnat}
\bibliography{references}  

\end{document}


\maketitle

\section{Proof of Theorem 1}
We provide a formal proof of Theorem 1 in the main text. From the main text, the posterior density using the joint LEAP is given by 
\begin{align}
    p(\bm{\theta}, \bm{\gamma}, \bm{c}_0 &| D, D_0)
       \propto \pi_0(\bm{\theta}, \bm{\gamma}) \left[ \left( \prod_{i=1}^{n} f(D_{i} | \bm{\theta}_1) \right) \left( \prod_{\{ c_{0i} = 1 \}} f(D_{0i} | \bm{\theta}_1) \right) \right] \gamma_1^{n_{01}}
       \notag \\
       &
       \times 
       \prod_{k=2}^K \left[ \gamma_k^{n_{0k}}
       \prod_{i=1}^{n_0} f(D_{0i} | \bm{\theta}_k)^{c_{0ik}}
       \right].
       \label{eq:leap_jointpost}
\end{align}
And the joint density of $(D_0, \bm{c}_0)$ is given by
\begin{align}
    f(D_0, \bm{c}_0 | \bm{\theta}, \bm{\gamma}) = \prod_{i=1}^{n_0} \prod_{k=1}^K \left[ \gamma_k f(D_{0i} | \bm{\theta}_k) \right]^{ c_{0ik}}.
    \label{eq:mixdens_latent}
\end{align}

\begin{theorem}
    Let $L(\bm{\theta}_1 | D)$ denote the likelihood of the current data set. If, for every partition $\{ 1, \ldots, K \}^{n_0}$, the following conditions hold, the posterior density $p(\bm{\theta} | D, D_0)$ in \eqref{eq:leap_jointpost} is proper
    \begin{enumerate}
        \item $\int 
        \pi_0(\bm{\theta}_1)
        \left[ \prod_{i=1}^n f(D_i | \bm{\theta}_1) \right]
        \left[ \prod_{\{ i : c_{0i1} = 1 \}}^{n_0} f(D_{0i} | \bm{\theta}_1) \right]
        d\bm{\theta}_1 < \infty$, and
        \item $\int \int \gamma_1^{n_{01}}
        \prod_{k=2}^K \gamma_k^{n_{0k}} \prod_{i=1}^{n_0} f(D_{0i} | \bm{\theta}_k)^{c_{0ik}}
        \pi_0( \bm{\theta}_{(-1)}, \bm{\gamma} | \bm{\theta}_1)
        d\bm{\theta}_{(-1)}d\bm{\gamma} < \infty$ for every $\bm{\theta}_1$.
    \end{enumerate}
    \label{thm:postproper}
\end{theorem}

\begin{proof}
Let $\tilde{p}$ denote the unnormalized posterior density under the LEAP, given by
\begin{align}
    \tilde{p}(\bm{\theta}, \bm{\gamma}, \bm{c}_0 | D, D_0) &= \pi_0(\bm{\theta}, \bm{\gamma})
    \left[\prod_{i=1}^{n} f(D_{0i} | \bm{\theta}_1)\right]
    \left[ \prod_{k=1}^K \gamma_k^{n_{0k}} \prod_{i=1}^{n_{0}} f(D_{0i} | \bm{\theta}_k)^{c_{0ik}}   \right]
    \label{eq:unnormalizedpost}
    .
\end{align}
The normalizing constant for \eqref{eq:unnormalizedpost} is thus given by
\begin{align}
    Z(D, D_0) &=  \sum_{ \bm{c}_0 \in \{ 1, \ldots, K \}^{n_0} } \int \int \int  
    \tilde{p}(\bm{\theta}, \bm{\gamma}, \bm{c}_0 | D, D_0)
      d\bm{\gamma} d\bm{\theta}_{(-1)} d\bm{\theta}_1
    , \notag \\
    &= 
    \sum_{ \bm{c}_0 \in \{ 1, \ldots, K \}^{n_0} }
    \int 
    \pi_0(\bm{\theta}_1)
    \left[ \prod_{i=1}^n f(D_i | \bm{\theta}_1) \right]
    \left[ \prod_{\{ i : c_{0i1} = 1 \}}^{n_0} f(D_{0i} | \bm{\theta}_1) \right]
    d\bm{\theta}_1
    \notag \\
    &\hspace{1cm} \times
    \int \int \gamma_1^{n_{01}}
    \prod_{k=2}^K \gamma_k^{n_{0k}} \prod_{i=1}^{n_0} f(D_{0i} | \bm{\theta}_k)^{c_{0ik}}
    \pi_0( \bm{\theta}_{(-1)}, \bm{\gamma} | \bm{\theta}_1)
    d\bm{\theta}_{(-1)}d\bm{\gamma},
    \label{eq:nc1}
\end{align}
where 
$\bm{\theta}_{(-1)} = (\bm{\theta}_2', \ldots, \bm{\theta}_K')'$ are the parameters excluding the first component,
$\pi_0(\bm{\theta}_1)$ is the marginal prior for $\bm{\theta}_1$,
and $\pi_0(\bm{\theta}_{(-1)}, \bm{\gamma} | \bm{\theta}_1)$ is the conditional prior for $\bm{\theta}_{(-1)}, \bm{\gamma}$ given $\bm{\theta}_1$.
Since the normalizing constant in \eqref{eq:nc1} involves a sum of nonnegative components, it is finite provided each component of the sum is finite.
It follows that $Z(D, D_0)$ is finite (and hence, the posterior under the LEAP is proper) if and only if for every $\bm{c}_0 \in \{ 1, \ldots, K \}^{n_0}$.
\end{proof}

For condition (1), note that, for certain models, $\pi_0(\bm{\theta}_1)$ need not be proper. Indeed, if $\int \pi_0(\bm{\theta}_1) d\bm{\theta}_1 = \infty$ but $\int \prod_{i=1}^n f(D_i | \bm{\theta}_1) \pi_0(\bm{\theta}_1) d\bm{\theta}_1 < \infty$, the posterior density under the LEAP is proper. Hence, under certain contexts (e.g., the normal linear model), we may elicit an improper initial prior for parameters pertaining to the current data model.

For condition (2), note that the marginal initial prior density for $\bm{\gamma}$ is defined over the $K$-dimensional simplex, so the marginal prior for $\bm{\gamma}$ is always proper. Since there exist partitions $\bm{c}_0^*$ such that $n_{0k} = 0$ for some $k$, it is clear that we require $\pi_0(\bm{\theta}_{(-1)})$ to be proper in order for the posterior to be proper whenever $\bm{\theta}_{(-1)}$ is unbounded.

\section{Posterior distribution under conjugate priors}
\label{sec:conjugate}
In this section, we show that, under a conjugate prior, the posterior distribution conditional on $\bm{c}_0$ has a simple closed form. Using this result, we develop an efficient Gibbs sampler to conduct posterior inference. 

Let $\pi_0(\bm{\theta}, \bm{\gamma}) = \pi_0(\bm{\gamma}) \prod_{k=1}^K \pi_0(\bm{\theta}_k)$ and suppose that $\pi_0(\bm{\theta}_k)$ is conjugate for the current data (e.g., if $D$ consists of i.i.d. Poisson observations, then $\pi_0(\bm{\theta}_k)$ is a gamma density). Let $g(\bm{\theta}_k | \bm{\eta}_k)$ denote this conjugate distribution, where $\bm{\eta}_k$ is a set of hyperparameters. Furthermore, let $\pi_0(\bm{\gamma} | \bm{\alpha}_0) \propto \prod_{k=1}^K \gamma_k^{\alpha_{0k} - 1}$, where $\bm{\alpha}_0 = (\alpha_1, \ldots, \alpha_{0k})'$ are concentration parameters and $\alpha_{0k} > 0$ for every $k \in \{1, \ldots, K\}$. Note that $\bm{\gamma}$ is a Dirichlet random vector a priori, which is conjugate to the categorical probability mass function (p.m.f.) in \eqref{eq:mixdens_latent}. It follows that the joint posterior distribution under the LEAP is given by
\begin{align}
    p(\bm{\theta}, \bm{\gamma}, \bm{c}_0 | D, D_0) \propto 
       g(\bm{\theta}_1 | D, D_{01}, \bm{\eta}_1)
       \left[ \prod_{k=1}^K \gamma_k^{n_{0k} + \alpha_{0k} - 1} \right]
       \left[ 
       \prod_{k=1}^K g(\bm{\theta}_k | D_{0k}, \bm{\eta}_k)
       \right],
       \label{eq:leap_post_conj}
\end{align}
where $D_{0k}$ denotes subset of the historical data classified into class $k$, $g(\bm{\theta}_k | D_{0k}, \bm{\eta}_k)$ conjugate prior updated with data $D_{0k}$, and $g(\bm{\theta}_1 | D, D_{01})$ is analogously defined.

Under independent conjugate priors, it is easy to sample from the posterior in \eqref{eq:leap_post_conj} using an efficient Gibbs sampler. In particular, we may sample $\bm{\theta}, \bm{\gamma} | \bm{c}_0$ and then $\bm{c}_0 | \bm{\theta}, \bm{\gamma}$. Suppose the current state is given by $\bm{\theta}^{(t)}$, $\bm{c}_0^{(t)}$. We may obtain a new sample via:
\begin{enumerate}
    \item Sample $\bm{\theta}_1^{(t+1)} \sim g(\cdot | D, D_{01}^{(t)}, \bm{\eta}_{1})$.
    \item For $k \in \{2, \ldots, K\}$, sample $\bm{\theta}_k^{(t+1)} \sim g(\cdot | D_{0k}^{(t)}, \bm{\eta}_k)$
    \item Sample $\bm{\gamma}^{(t+1)} \sim \text{Dir}(\bm{n}_0^{(t)} + \bm{\alpha}_0)$
    \item For $i \in \{1, \ldots, n_0\}$, sample $c_{0i}^{(t+1)} \sim \text{Categorical}(p_{i1}^{(t+1)}, \ldots, p_{iK}^{(t+1)})$, where
    $$
    p_{ik}^{(t+1)} \propto \gamma_k^{(t+1)} f(D_{0i} | \bm{\theta}_k^{(t+1)})
    $$
\end{enumerate}
where we let $\bm{n}_0^{(t)} = (n_{01}^{(t)}, \ldots, n_{0K}^{(t)})'$ denote the number of subjects assigned to each class at the $t^{th}$ iteration.

It is worth noting that MCMC samplers for mixture models suffer from the dreaded ``label switching problem,'' which is caused when the posterior density for a parameter can switch indices with another parameter to yield an equivalent likelihood. However, since we are only concerned with the posterior density of $\bm{\theta}_1$, the problem is avoided since the current data set is \emph{always} parameterized by $\bm{\theta}_1$, so that the first component of the mixture model in \eqref{eq:mixdens_latent} always shares its parameter with the current data.

\section{The partially truncated Dirichlet distribution (PTDD)}
\subsection{Definition of the distribution}
Let $\bm{X} = (X_1, \ldots, X_K)'$ be a $K$-dimensional Dirichlet vector with concentration parameter $\bm{\alpha} = (\alpha_1, \ldots, \alpha_K)$. Then the pdf of $\bm{X}$ is given by
$$
  f_{\bm{X}}(\bm{x}) = \frac{1}{B(\bm{\alpha})} \prod_{k=1}^K x_k^{\alpha_k - 1}.
$$

Consider the vector $\bm{Y} = (Y_1, \ldots, Y_K)'$ defined as $Y_1 = X_1 1\{ a < X_1 < b \}$, where $0 \le a < b \le 1$, and $Y_k = X_k$ with $(Y_1, \ldots, Y_K) \in [0,1]^K$. Let $\Omega = \{ \bm{x} \in [0,1]^K : a \le x_1 \le b \}$ The joint density is given by
\begin{align}
    f_{\bm{Y}}(\bm{y}) = \frac{\prod_{k=1}^K y_k^{\alpha_k - 1}}{B(\bm{\alpha} ; a, b)} \text{ for } \bm{y} \in \Omega,
    \label{eq:ptdd}
\end{align}
where $B(\bm{\alpha} ; a, b) = \int_{\Omega} \prod_{k=1}^K y_k^{\alpha_k - 1} d\bm{y}$. We refer to the distribution with pdf given by \eqref{eq:ptdd} as the partially truncated Dirichlet distribution (PTDD). 

\subsection{Aggregation property}
We wish to find the marginal distribution of $Y_1$ and the conditional distribution $\bm{Y}_2 | Y_1$. To that end, consider the $(K-1)$-dimensional vector $\bm{Z} = (Z_1, \bm{Z}_2, Z_3)'$, where $Z_1 = Y_1$, $\bm{Z}_2 = (Y_2, \ldots, Y_{K-2})'$, and $Z_3 = (Y_{K-1} + Y_K)'$. We have
\begin{align}
  f_{\bm{Z}}(\bm{z}) &\propto \int_0^{z_3} z_1^{\alpha_1 - 1} 1\{ a < z_1 < b \} \left[ \prod_{k=2}^{K-2} z_k^{\alpha_k - 1} \right] u^{\alpha_{K-1} - 1} (z_3 - u)^{\alpha_K - 1} du
  \notag \\
  &=  z_1^{\alpha_1 - 1} 1\{ a < z_1 < b \} \left[ \prod_{k=2}^{K-2} z_k^{\alpha_k - 1} \right]
      \int_{0}^1 (z_3 v)^{\alpha_{K-1} - 1} (z_3 - z_3 v)^{\alpha_K - 1} z_3 dv 
  \notag \\
  &= z_1^{\alpha_1 - 1} 1\{ a < z_1 < b \} \left[ \prod_{k=2}^{K-2} z_k^{\alpha_k - 1} \right] z_3^{\alpha_{K-1} + \alpha_{K} - 1}
  \int_{0}^1 v^{\alpha_{K-1} - 1}(1 - v)^{\alpha_K - 1}
  \notag \\
  &\propto z_1^{\alpha_1 - 1} 1\{ a < z_1 < b \} \left[ \prod_{k=2}^{K-2} z_k^{\alpha_k - 1} \right] z_3^{\alpha_{K-1} + \alpha_{K} - 1},
  \label{eq:aggregation}
\end{align}
where, in the second line, we used the substitution $u = z_3 v$. It follows from \eqref{eq:aggregation} and \eqref{eq:ptdd} that $\bm{Z}$ follows a PTDD, i.e.,
$$
  (Z_1, \ldots, Z_{K-1}) \sim \text{PTDD}\left( \alpha_1, \ldots, \alpha_{K-2}, \alpha_{K-1} + \alpha_K; a, b \right).
$$
For the traditional Dirichlet distribution, this property is known as the ``aggregation property.'' 

\subsection{Marginal distribution of the truncated component}
Let $\bm{Z} = (Y_1, Y_2 + \ldots + Y_K)'$ Note that $\bm{Z} = (Y_1, 1 - Y_1)$, which only involves the first component of $\bm{Y}$. It follows from the aggregation property in \eqref{eq:aggregation} that we may write the marginal pdf of $Y_1$ as
\begin{align}
    f_{Y_1}(y) \propto y^{\alpha_1 - 1} (1 - y)^{\sum_{k=2}^K \alpha_k - 1} 1\{ a < y < b \},
    \notag
\end{align}
which is a truncated beta distribution, i.e., the normalized PDF of $Y_1$ may be written as
\begin{align}
    f_{Y_1}(y) = = \frac{y^{\alpha_1 - 1} (1 - y)^{\sum_{k=2}^K \alpha_k - 1}}{B\left( \alpha_1, \sum_{k=2}^K \alpha_k; a, b \right)} \text{ for } a < y < b,
\end{align}
where $B\left( \alpha, \beta; a, b \right) = \int_{a}^b y^{\alpha - 1}(1 - y)^{\beta - 1} dy$ is easily computed using the CDF of a beta random variable, i.e.,
$$
  B\left( \alpha, \beta; a, b \right) = F_{\text{Beta}}(b | \alpha, \beta) - F_{\text{Beta}}(a | \alpha, \beta).
$$

\subsection{Conditional distribution}
We now derive the conditional distribution of $\bm{Y}_2 | Y_1$. Let $\alpha_0 = \sum_{k=2}^K \alpha_k$. We have
\begin{align}
  f(\bm{y}_2 | y_1) &= \frac{f(y_1, \bm{y}_2)}{f(y_1)} = \frac{ \frac{\prod_{k=1}^K y_k^{\alpha_k - 1}}{B(\bm{\alpha} ; a, b)} }{\frac{y_1^{\alpha_1 - 1} (1 - y_1)^{\sum_{k=2}^K \alpha_k - 1}}{B\left( \alpha_1, \sum_{k=2}^K \alpha_k; a, b \right)} }
  \notag \\
  &= \frac{B\left( \alpha_1, \sum_{k=2}^K \alpha_k; a, b \right)}{B(\bm{\alpha} ; a, b)}
     \left[ \prod_{k=2}^K y_k^{\alpha_k - 1} \right] 
     (1 - y_1)^{-(\alpha_0 - 1)}
  \notag \\
  &= \frac{B\left( \alpha_1, \sum_{k=2}^K \alpha_k; a, b \right)}{B(\bm{\alpha} ; a, b)}
       \left[ \prod_{k=2}^K y_k^{\alpha_k - 1} (1 - y_1)^{-\alpha_k} \right] (1 - y_1)
  \notag \\
  &= \frac{B\left( \alpha_1, \sum_{k=2}^K \alpha_k; a, b \right)}{B(\bm{\alpha} ; a, b)}
     \left[ \prod_{k=2}^K \left(\frac{y_k}{1 - y_1}\right)^{\alpha_k - 1} \right] (1 - y_1)^{-[(K-1) - 1]}
  \notag \\
  &\propto \prod_{k=2}^K \left( \frac{y_k}{1 - y_1} \right)^{\alpha_k - 1}.
  \label{eq:cond_dens}
\end{align}
It follows from \eqref{eq:cond_dens} that
\begin{align}
    \frac{\bm{Y}_2}{1 - y_1} \sim \text{Dirichlet}\left( \alpha_2, \ldots, \alpha_K \right),
\end{align}
i.e., $\bm{Y}_2 | Y_1 = y_1$ has a scaled Dirichlet distribution, which we may write as
$$
  \bm{Y}_2 | Y_1 = y_1 \sim (1 - y_1) \text{Dirichlet}(\alpha_2, \ldots, \alpha_K).
$$

It follows that, for $\bm{y} \in \Omega$, the joint density is given by
\begin{align}
    f(\bm{y}) &= f(y_1) f(\bm{y}_2 | y_1) \notag \\
    &= \frac{1}{ B(\alpha_1, \alpha_0, a, b) B(\bm{\alpha}_{(-1)})}  
         y_1^{\alpha_1 - 1}(1 - y_1)^{\alpha_0 - 1} \left[\prod_{k=2}^K \left( \frac{y_k}{1 - y_1}  \right)^{\alpha_k - 1} \right] (1 -y_1)^{-(K-1)}
    \notag \\
    &= \frac{1}{ B(\alpha_1, \alpha_0, a, b) B(\bm{\alpha}_{(-1)})}  
         \prod_{k=1}^K y_k^{\alpha_k - 1},
\end{align}
where $\bm{\alpha}_{(-1)} = (\alpha_2, \ldots, \alpha_K)'$. Hence, the PTD density is computable using software implementations of the multivariate beta function and the CDF of the beta density.

\subsection{Random number generation}
It is quite easy to generate random numbers from the PTDD. We may use the following algorithm
\begin{enumerate}
    \item Generate $y_1 \sim \text{TruncBeta}(\alpha_1, \alpha_0; a, b)$
    \item Generate $\tilde{\bm{y}}_2 \sim \text{Dirichlet}(\alpha_2, \ldots, \alpha_K)$ and set $\bm{y}_2 = (1 - y_1) \tilde{\bm{y}}_2$.
\end{enumerate}
For (1), we may generate truncated beta random variables using rejection sampling or the inverse-CDF method, recognizing that the CDF of $Y_1$ is given by 
$$
  F_{Y_1}(y) = \frac{F_{\text{Beta}}(y | \alpha_1, \alpha_0) - F_{\text{Beta}}(a | \alpha_1, \alpha_0)}{F_{\text{Beta}}(b | \alpha_1, \alpha_0) - F_{\text{Beta}}(a | \alpha_1, \alpha_0) },
$$
and hence we can generate truncated beta random variables directly via
\begin{enumerate}
    \item Generate $u_1 \sim \text{Unif}(F_{\text{Beta}}(a | \alpha_1, \alpha_0), F_{\text{Beta}}(b | \alpha_1, \alpha_0))$
    \item Set $y_1 = F_{\text{Beta}}^{-1}\left( u_1 \right)$
\end{enumerate}

\subsection{Moments of the PTD Distribution}
Expectations of the PTD can be computed in closed-form. Specifically, if 
$\bm{Y}= (Y_1, \ldots, Y_K) \sim \text{PTD}(\bm{\alpha}, a, b)$, then
\begin{align}
    E\left[ \prod_{k=1}^K Y_k^{m_k} \right] 
      &= \frac{1}{B\left(\alpha_1, \alpha_0, a, b\right)B\left(\bm{\alpha}_{(-1)}\right)} \int_{\Omega} \prod_{k=1}^K y_k^{m_k + \alpha_k - 1} d\bm{y},
      \notag\\
      &= \frac{ B\left(\alpha_1 + m_1, \alpha_0 + m_0, a, b \right) B\left (\bm{\alpha}_{(-1)} + \bm{m}_{(-1)} \right) }{ B\left(\alpha_1, \alpha_0, a, b)B(\bm{\alpha}_{(-1)}\right) },
      \label{eq:ptd_moments}
\end{align}
where $\alpha_0 = \sum_{k=2}^K \alpha_k$, $m_0 = \sum_{k=2}^K m_k$ and $\bm{m}_{(-1)} = (m_2, \ldots, m_K)'$. In particular, the expression in \eqref{eq:ptd_moments} indicates that
$$
  E[X_1] = \frac{ B\left(\alpha_1 + 1, \alpha_0, a, b \right) }{ B\left(\alpha_1, \alpha_0, a, b \right) },
$$
from which we recover $E[X_1] = \alpha_1 / (\alpha_1 + \alpha_0)$ when $(m_2, \ldots, m_K) = \bm{0}$, $a = 0$, and $ b= 1$, which is the mean of an untruncated beta distribution.

\subsection{Conjugacy of the PTBB}
Suppose that $\bm{x} | \bm{\gamma} \sim \text{Multinomial}(\bm{n}, \bm{\gamma})$, where $\bm{n} = (n_1, \ldots, n_K)$ is known and $\bm{\gamma} = (\gamma_1, \ldots, \gamma_K)'$. Suppose, a priori, that $\bm{\gamma} \sim \text{PTBB}\left( \bm{\alpha}, a, b \right)$. Then
\begin{align}
    p(\bm{\gamma} | \bm{x}) \propto 1\{ a < \gamma_k < b \} \prod_{k=1}^K \gamma_k^{n_k + \alpha_k - 1},
\end{align}
which is the kernel of a PTBB random variable with concentration parameter $\bm{n} + \bm{\alpha}$ and truncation parameters $a$ and $b$. It follows that the PTBB is conjugate to the multinomial distribution.

\section{LEAP for the normal linear model}
In this section, we develop the LEAP for the normal linear model, providing detailed derivations.

For the linear model, the full LEAP is given by
\begin{align}
    \pi( \bm{\theta}, &\bm{c}_0, \bm{\gamma} | D_0) \propto \prod_{i=1}^{n_0} \prod_{k=1}^K \left[ \gamma_k \phi(y_{0i} | \bm{x}_i'\bm{\beta}_k, \tau_k^{-1}) \right]^{c_{0ik}} \pi_0(\bm{\theta}, \bm{\gamma}),
    \notag \\
    &\propto \prod_{k=1}^K \gamma_k^{n_{0k}} 
    \tau_k^{n_{0k} / 2} 
    \exp\left\{ -\frac{\tau_k}{2} \sum_{i=1}^{n_0} c_{0ik} (\bm{y}_{0i} - \bm{x}_{0i}'\bm{\beta}_k)^2 \right\} \pi_0(\bm{\theta}, \bm{\gamma}),
    \notag \\
    &\propto \prod_{k=1}^K \gamma_k^{n_{0k}} \tau_k^{n_{0k} / 2} \exp\left\{ -\frac{\tau_k}{2} \left( \bm{y}_0 - \bm{X}_0 \bm{\beta}_k \right)' \bm{C}_{0k} \left( \bm{y}_0 - \bm{X}_0 \bm{\beta}_k \right) \right\} 
    \pi_0(\bm{\theta}, \bm{\gamma}),
    \label{eq:lm_leap_likelihood_hist}
\end{align}
where $\bm{C}_{0k} = \text{diag}\left\{ c_{0ik}, i = 1, \ldots, n_0 \right\}$ is a diagonal matrix of 0's and 1's indicating whether an individual belongs to class $k$ and $\sum_{k=1}^K \bm{C}_{0k} = \bm{I}_n$. Note that, ignoring $\bm{\gamma}$ and the initial prior $\pi_0(\bm{\theta}, \bm{\gamma})$, the prior in \eqref{eq:lm_leap_likelihood_hist} is analogous to a product of $K$ independent likelihoods from a normal linear model, each with its own regression coefficient and precision parameter. In other words, we assume that the historical data come from a Gaussian mixture model (GMM). We may thus use properties of the Bayesian linear model to compute the joint induced prior from the mixture model.

Note that $\bm{C}_{0k}$ is idempotent, i.e., $\bm{C}_{0k} \bm{C}_{0k} = \bm{C}_{0k}$. Let $\bm{y}_{0k} = \bm{C}_{0k} \bm{y}_0$, $\bm{X}_{0k} = \bm{C}_{0k} \bm{X}_0$, and 
$$
\bm{M}_{0k} = \bm{X}_{0k} \left( \bm{X}_{0k}' \bm{X}_{0k} \right)^{-}\bm{X}_{0k}' 
  = \bm{C}_{0k} \bm{X}_{0}(\bm{X}_{0}'\bm{C}_{0k} \bm{X}_0)^{-1} \bm{X}_0'\bm{C}_{0k},
$$
where $\bm{A}^{-}$ denotes a generalized inverse of $\bm{A}$.
Further, let 
$$
    \hat{\bm{\beta}}_{0k} = (\bm{X}_{0k}'\bm{X}_{0k})^{-} \bm{X}_{0k}'\bm{y}_{0k} = (\bm{X}_{0}'\bm{C}_{0k}\bm{X}_{0})^{-}\bm{X}_0'\bm{C}_{0k}\bm{y}_0
    .
$$
Let the initial prior be given by
\begin{align}
    \pi_0(\bm{\theta}, \bm{\gamma}) \propto \prod_{k=1}^K \gamma_k^{\alpha_{0k} - 1} \tau^{\delta_{0k} / 2 - 1} \exp\left\{ -\frac{\tau_k \xi_{0k}}{2} \right\} \tau_k^{p/2} \exp\left\{ \left( \bm{\beta}_k - \bm{\mu}_{0k} \right)' \bm{\Omega}_{0k}^{-1} \left( \bm{\beta}_k - \bm{\mu}_{0k} \right) \right\}
    ,
    \label{lm_initprior}
\end{align}
i.e., the initial prior is a product of $K$ independent priors in the normal-gamma family. Thus,
\begin{align}
  \pi( \bm{\theta}, &\bm{c}_0, \bm{\gamma} | D_0) \propto
     \prod_{k=1}^K \gamma_k^{n_{0k} + \alpha_{0k} + p - 1} \tau_k^{(n_{0k} + \delta_{0k}) / 2 - 1}
    \exp\left\{ -\frac{\tau_k}{2} \left[ \xi_{0k} + \bm{y}_{0k}'(I_{n_0} - \bm{M}_{0k}) \bm{y}_{0k} \right] \right\}
    \notag \\
    &\hspace{0.5cm} \times 
    \exp\left\{ -\frac{\tau_k}{2} \left( \bm{\beta}_k - \hat{\bm{\beta}}_{0k} \right)' \bm{X}_0'\bm{C}_{0k}\bm{X}_0 \left( \bm{\beta}_k - \hat{\bm{\beta}}_{0k} \right)
    + \left( \bm{\beta}_k - \bm{\mu}_{0k} \right)' \bm{\Omega}_{0k} \left( \bm{\beta}_k - \bm{\mu}_{0k} \right)
    \right\}.
    \label{eq:lm_leap_conj1}
\end{align}
Let $\tilde{\bm{\beta}}_{0k} = \bm{\Lambda}_{0k} \bm{\mu}_{0k} + (\bm{I}_p - \bm{\Lambda}_{0k}) \hat{\bm{\beta}}_{0k}$, where $\bm{\Lambda}_{0k} = \left( \bm{X}_0'\bm{C}_{0k} \bm{X}_0 + \bm{\Omega}_{0k} \right)^{-1} \bm{\Omega}_{0k}$. Note that $\tilde{\bm{\beta}}_{0k}$ is a convex combination (weighted average) of the MLE of class $k$ and the prior for $\bm{\beta}_k$. Moreover, let 
$\tilde{s}_{0k}^2 = \bm{y}_{0k}'\left( I_{n_0} - \bm{M}_{0k} \right) \bm{y}_{0k} + (\hat{\bm{\beta}}_{0k} - \bm{\mu}_{0k})' (\bm{\Lambda}_{0k}' \bm{X}_{0}'\bm{C}_{0k} \bm{X}_0) (\hat{\bm{\beta}}_{0k} - \bm{\mu}_{0k}) + \xi_{0k}$. It follows from properties of the Bayesian linear model that the joint leap may be expressed as
\begin{align}
      &\pi( \bm{\theta}, \bm{c}_0, \bm{\gamma} | D_0) \propto
     \prod_{k=1}^K 
     \left\lvert \bm{X}_0'\bm{C}_{0k} \bm{X}_0 + \bm{\Omega}_{0k} \right\rvert^{-1/2}
     \frac{\Gamma\left( \frac{n_{0k} + \delta_{0k}}{2} \right)}{\left( \tilde{s}_{0k}^2 / 2 \right)^{(n_{0k} + \alpha_{0k})/2}}
     \notag \\
     &\hspace{0.5cm} \times
     \gamma_k^{n_{0k} + \alpha_{0k} - 1}
         \frac{\left( \tilde{s}_{0k}^2 / 2 \right)^{(n_{0k} + \alpha_{0k})/2}}{\Gamma\left( \frac{n_{0k} + \delta_{0k}}{2} \right)}
     \tau_k^{(n_{0k} + \delta_{0k}) / 2 - 1}
    \exp\left\{ -\frac{\tau_k}{2} \tilde{s}^2_{0k} \right\}
    \notag \\
    &\hspace{0.5cm} \times 
    \tau_k^{p/2} 
         \left\lvert \bm{X}_0'\bm{C}_{0k} \bm{X}_0 + \bm{\Omega}_{0k} \right\rvert^{1/2}
    \exp\left\{ -\frac{\tau_k}{2} \left( \bm{\beta}_k - \tilde{\bm{\beta}}_{0k} \right)' \left( \bm{X}_0'\bm{C}_{0k}\bm{X}_0 + \bm{\Omega}_{0k} \right) \left( \bm{\beta}_k - \tilde{\bm{\beta}}_{0k} \right)
    \right\}.
    \label{eq:lm_leap_conj2}
\end{align}
We may express the prior in \eqref{eq:lm_leap_conj2} hierarchically as
\begin{align}
    \bm{\beta}_k | \tau_k, \bm{c}_0 &\sim N_p\left( \tilde{\bm{\beta}}_{0k}, \tau_k^{-1} \left[ \bm{X}_0'\bm{C}_{0k} \bm{X}_0 + \bm{\Omega}_{0k} \right]^{-1} \right),
    \notag \\
    \tau_k | \bm{c}_0 &\sim \text{Gamma}\left( \frac{n_{0k} + \delta_{0k}}{2}, \frac{\tilde{s}_{0k}^2}{2} \right),
    \notag \\
    \bm{\gamma} | \bm{c}_0 &\sim \text{Dirichlet}\left( \bm{n}_0 + \bm{\alpha}_0 \right),
    \notag \\
    \bm{c}_0 &\sim H_0,
    \label{leap_conj_hierarchy}
\end{align}
where $\bm{n}_0 = (n_{01}, \ldots, n_{0K})'$ contains the sizes of each class, $H_0$ is the p.m.f. induced by the prior in \eqref{eq:lm_leap_conj2}. It is clear then that the LEAP is a proper prior if and only if $\bm{\Omega}_{0k}$ is a positive definite matrix, $\delta_{0k} > 0$ for every $k$, and $\xi_{0k} > 0$ for every $k$. Equivalently, the LEAP is proper provided the initial prior $\pi_0(\bm{\theta}, \bm{\gamma})$ is proper.

We may thus obtain the marginal prior for $\bm{\theta}_1$ as
\begin{align}
    \pi(\bm{\theta}_1 | D_0) &= \sum_{\bm{c}_0 \in \{ 1, \ldots, K \}^{n_0}} \phi_p\left( \bm{\beta}_1 \left| \tilde{\bm{\beta}}_{01}, \tau_1^{-1} \left[ \bm{X}_0'\bm{C}_{0k}\bm{X}_0 + \bm{\Omega}_{0k} \right]^{-1} \right) \right.
    f_{\Gamma}\left( \tau_1 \left| \frac{n_{0k} + \delta_{0k}}{2} , \frac{\tilde{s}_{0k}^2}{2} \right) \right.
    f_{H_0}(\bm{c}_0),
     \label{eq:lm_leap_marginal_prior}
\end{align}
where $\phi_p(\cdot | \bm{m}, \bm{C})$ is the $p$-dimensional multivariate normal density function with mean $\bm{m}$ and covariance matrix $\bm{C}$, $f_{\Gamma}(\cdot | a, b)$ is the gamma density function with shape parameter $a$ and rate parameter $b$, and $f_{H_0}(\cdot)$ is the marginal mass function of $\bm{c}_0$ given by
\begin{align}
  f_{H_0}(\bm{c}_0) = 
  \frac
  {B(\bm{n}_0 + \bm{\alpha}_0) \prod_{k=1}^K 
     \left\lvert \bm{X}_0'\bm{C}_{0k} \bm{X}_0 + \bm{\Omega}_{0k} \right\rvert^{-1/2}
     \frac{\Gamma\left( \frac{n_{0k} + \delta_{0k}}{2} \right)}{\left( \tilde{s}_{0k}^2 / 2 \right)^{(n_{0k} + \alpha_{0k})/2}}}
  {\sum\limits_{c_0 \in \{ 1, \ldots, K \}^{n_0}}
    B(\bm{n}_0 + \bm{\alpha}_0) \prod_{k=1}^K 
     \left\lvert \bm{X}_0'\bm{C}_{0k} \bm{X}_0 + \bm{\Omega}_{0k} \right\rvert^{-1/2}
     \frac{\Gamma\left( \frac{n_{0k} + \delta_{0k}}{2} \right)}{\left( \tilde{s}_{0k}^2 / 2 \right)^{(n_{0k} + \alpha_{0k})/2}}
  }
  ,
  \label{eq:marginal_prior_c0}
\end{align}
where $B(\cdot)$ is the multivariate beta function and $\bm{n}_0 = (n_{01}, \ldots, n_{0K})$ is the sample size for each class. 

It follows from \eqref{eq:lm_leap_marginal_prior} that the LEAP is a mixture of $K^{n_0}$ normal-gamma densities, where the mixture weights are given by $f_{H_0}(\cdot)$. Although the denominator (and hence, the mixture weights) in \eqref{eq:marginal_prior_c0} is not in general computationally tractable, the form of the marginal prior for $\bm{c}_0$ is nevertheless helpful for understanding the LEAP. First, note that the kernel for $f_{H_0}(\cdot)$ is decreasing in $\tilde{s}_{0k}^2$. We may take $\bm{\Omega}_1$ to be noninformative, e.g., $\bm{\Omega}_1 = \text{diag}\{ 10^{-2}, \ldots, 10^{-2} \}$ so that the prior mean is not heavily influenced by the initial prior mean. Second, all else being equal, the marginal mass function for $\bm{c}_0$ prefers partitions of $\{ 1, \ldots, K\}^{n_0}$ that yield smaller sums of squared residuals within each class. Hence, for equal values of $n_{01} \in \{1, 2, \ldots, n_0-1\}$, the marginal prior density for $\bm{\theta}_1$ is more influenced by partitions that fit well to a linear regression model.

\subsection{The posterior density}
\label{sec:leap_lm_posterior}
In this section, we present analytic results concerning the posterior distribution under the LEAP. We also discuss some of the borrowing properties of the LEAP. For brevity, we summarize the analytic results. Full derivations are provided in Section 4 of the Supplementary Appendix.

Let $D = \{(y_i, \bm{x}_i), i = 1, \ldots, n\}$ denote the current data set of sample size $n$, where $y_i$ is a scalar response variable and $\bm{x}_i$ is a $p$-dimensional vector of covariates. We assume a linear regression model of the form 
$$
  \bm{y} = \bm{X}\bm{\beta}_1 + \bm{\epsilon},
$$
where 
$\bm{y} = (y_1, \ldots, y_n)'$ is an $n$-dimensional response vector,
$\bm{X} = (\bm{x}_1, \ldots, \bm{x}_n)'$ is a full column rank $n\times p$ design matrix, 
$\bm{\beta}_1$ is a $p$-dimensional vector of regression coefficients, 
and $\bm{\epsilon} = (\epsilon_1, \ldots, \epsilon_n)'$ is an $n$-dimensional error term. We assume $\bm{\epsilon}\sim N(0, \tau_1^{-1} I_n)$. We may write the likelihood of the current data $D$ as
\begin{align}
    L(\bm{\theta}_1 | D) \propto \tau_1^{-n/2} 
    \exp\left\{ -\frac{\tau_1}{2} \bm{y}'(I_n - \bm{M}) \bm{y} \right\}
    \exp\left\{ -\frac{\tau_1}{2} \left( \bm{\beta}_1 - \hat{\bm{\beta}}_1 \right)' \bm{X}'\bm{X}
    \left( \bm{\beta}_1 - \hat{\bm{\beta}}_1 \right)
    \right\},
    \label{eq:curlike}
\end{align}
where $\bm{M} = \bm{X}\left( \bm{X}'\bm{X} \right)^{-1} \bm{X}'$ is the orthogonal projection operator onto the space spanned by the columns of $\bm{X}$ and where $\hat{\bm{\beta}} = \left( \bm{X}'\bm{X} \right)^{-1} \bm{X}'\bm{y}$ is the MLE of the current data set. 

For notational convenience, let $\bm{\Psi}_{0k} = \bm{X}_{0}'\bm{C}_{0k} \bm{X}_0 + \bm{\Omega}_{0k}$ denote the precision matrix for the $k^{th}$ class. 
Then the joint posterior under the joint LEAP is given by
\begin{align}
    p(\bm{\theta}, \bm{\gamma}, &\bm{c}_0 | D, D_0) 
      \propto L(\bm{\theta}_1 | D) \pi(\bm{\theta}, \bm{\gamma}, \bm{c}_0 | D_0),
      \notag \\
      &\propto 
        \tau_1^{n/2} \exp\left\{ -\frac{\tau_1}{2} \bm{y}'\left( I_n - \bm{M} \right) \bm{y} \right\}
        \exp\left\{ -\frac{\tau_1}{2} \left( \bm{\beta}_1 - \hat{\bm{\beta}}_1 \right)' \bm{X}'\bm{X} \left( \bm{\beta}_1 - \hat{\bm{\beta}}_1 \right) \right\}
      \notag \\
      &\hspace{0.5cm} \times 
      \prod_{k=1}^K \gamma_k^{n_{0k} + \alpha_{0k} + p - 1}
      \tau_k^{(n_{0k} + \delta_{0k}) / 2 - 1}\exp\left\{ -\frac{\tau_k}{2} \tilde{s}_{0k}^2 \right\}
      \exp\left\{ -\frac{\tau_k}{2} 
        \left( \bm{\beta}_k - \tilde{\bm{\beta}}_{0k} \right)' \bm{\Psi}_{0k} \left( \bm{\beta}_k - \tilde{\bm{\beta}}_{0k} \right)
      \right\},
      \notag \\
      &\propto \left[ \prod_{k=1}^K \gamma_k^{n_{0k} + \alpha_{0k} - 1} \right]
      \left[
      \prod_{k=2}^K 
        \tau_k^{(n_{0k} + \delta_{0k} + p) / 2 - 1}\exp\left\{ -\frac{\tau_k}{2} \tilde{s}_{0k}^2 \right\}
       \exp\left\{ -\frac{\tau_k}{2} 
        \left( \bm{\beta}_k - \tilde{\bm{\beta}}_{0k} \right)' \bm{\Psi}_{0k} \left( \bm{\beta}_k - \tilde{\bm{\beta}}_{0k} \right)
      \right\}
      \right]
    \notag \\
    &\hspace{0.5cm} \times 
      \tau_1^{(n + n_{01} + \delta_{01} + p)/2}\exp\left\{ -\frac{\tau_1}{2}\left( \bm{y}'(I_n - \bm{M}) \bm{y} + \tilde{s}_{01}^2 \right) \right\}
      \notag \\ &\hspace{0.5cm} \times 
      \exp\left\{
        -\frac{\tau_1}{2}
        \left[
          \left(\bm{\beta}_1 - \hat{\bm{\beta}}_1 \right)' \bm{X}'\bm{X} \left( \bm{\beta}_1 - \hat{\bm{\beta}}_1 \right)
          + \left(\bm{\beta}_1 - \tilde{\bm{\beta}}_{01} \right)' \bm{\Psi}_{01} \left(\bm{\beta}_1 - \tilde{\bm{\beta}}_{01}\right)
        \right]
      \right\}.
      \label{eq:lm_post1}
\end{align}
Let 
$\bm{\Lambda}_1 = \left( \bm{X}'\bm{X} + \bm{\Psi}_{01} \right)^{-1} \bm{\Psi}_{01}$, 
$\tilde{\bm{\beta}}_1 = \bm{\Lambda}_1 \tilde{\bm{\beta}}_{01} + \left(\bm{I}_p - \bm{\Lambda}_1\right) \hat{\bm{\beta}}_1$, 
and 
$\tilde{s}_1^2 = \bm{y}'\left(\bm{I}_p - \bm{M}\right) \bm{y} + \left( \hat{\bm{\beta}}_1 - \tilde{\bm{\beta}}_{01} \right)'\left( \bm{\Lambda}_1 \bm{X}'\bm{X} \right) \left(\hat{\bm{\beta}}_1 - \tilde{\bm{\beta}}_{01}\right) + \xi_{01}$. 
It follows from a direct analogy of the derivation of the LEAP that the joint posterior density is given by
\begin{align}
    p(\bm{\theta}, &\bm{\gamma}, \bm{c}_0 | D, D_0) \propto 
      \left[ \prod_{k=1}^K \gamma_k^{n_{0k} + \alpha_{0k} - 1} \right]
      \notag \\
      &\hspace{0.5cm} \times
      \left[
      \prod_{k=2}^K 
        \tau_k^{(n_{0k} + \delta_{0k} + p) / 2 - 1}\exp\left\{ -\frac{\tau_k}{2} \tilde{s}_{0k}^2 \right\}
       \exp\left\{ -\frac{\tau_k}{2} 
        \left( \bm{\beta}_k - \tilde{\bm{\beta}}_{0k} \right)' \bm{\Psi}_{0k} \left( \bm{\beta}_k - \tilde{\bm{\beta}}_{0k} \right)
      \right\}
      \right]
      \notag \\ &\hspace{0.5cm} \times
      \tau^{(n + n_{01} + \delta_{01} + p)/2}
      \exp\left\{
        -\frac{\tau_1\tilde{s}_1^2}{2}
      \right\} 
      \exp\left\{ 
         -\frac{\tau_1}{2} \left[ \left( \bm{\beta} - \tilde{\bm{\beta}}_1 \right)' \left( \bm{X}'\bm{X} + \bm{\Psi}_{01} \right) \left( \bm{\beta} - \tilde{\bm{\beta}}_1 \right) \right]
      \right\},
    \notag \\
    &\propto B(\bm{n}_0 + \bm{\alpha}_0) 
    \frac{ \Gamma\left( \frac{n + n_{01} + \delta_{01}}{2} \right) }{ (\tilde{s}_1^2 / 2)^{(n + n_{01} + \delta_{01})/2}} \left\lvert \bm{X}'\bm{X} + \bm{\Psi}_{01} \right\rvert^{-1/2} 
      \left[\prod_{k=2}^K \frac{\Gamma\left( \frac{n_{0k} + \delta_{0k}}{2} \right)}{{(\tilde{s}_{0k}^2/2)}^{(n_{0k} + \delta_{0k})/2}}
      \left\lvert \bm{\Psi}_{0k} \right\rvert^{-1/2}
      \right]
      \notag \\ &\hspace{0.5cm} \times
        f_{\text{Dir}}\left( \bm{\gamma} | \bm{n}_0 + \bm{\alpha}_0 \right)
      \notag \\ &\hspace{0.5cm} \times
      \left[ \prod_{k=2}^K
        f_{\Gamma}\left( \tau_k \left| \frac{n_{0k} + \delta_{0k}}{2}, \frac{\tilde{s}_{0k}^2}{2} \right) \right. 
        \phi_p\left( \bm{\beta}_k \left| \tilde{\bm{\beta}}_{0k}, \tau_k^{-1} \bm{\Psi}_{0k}^{-1} \right) \right.
      \right]
   \notag \\ &\hspace{0.5cm} \times
   f_{\Gamma}\left( \tau_1 \left| \frac{n + n_{01} + \delta_{01}}{2}, \frac{\tilde{s}_1^2}{2} \right) \right.
   \phi_p\left( \bm{\beta}_1 \left| \tilde{\bm{\beta}}_1, \tau_1^{-1} \left[ \bm{X}'\bm{X} + \bm{\Psi}_{01} \right]^{-1} \right) \right.
   .
   \label{eq:lm_leap_jointpost}
\end{align}
From \eqref{eq:lm_leap_jointpost}, the marginal posterior density of $\bm{\theta}_1$ may be expressed as
\begin{align}
    p(\bm{\theta}_1 | D, D_0)
       \propto 
       \sum_{\bm{c}_0 \in \{ 1, \ldots, K \}^{n_0}}
       f_{\Gamma}\left( \tau_1 \left| \frac{n + n_{01} + \delta_{01}}{2}, \frac{\tilde{s}_1^2}{2} \right) \right.
   \phi_p\left( \bm{\beta}_1 \left| \tilde{\bm{\beta}}_1, \tau_1^{-1} \right) \right.
   f_H(\bm{c}_0),
   \label{eq:lm_posterior_theta1}
\end{align}
where
\begin{align}
  f_H(\bm{c}_0) = \frac{B(\bm{n}_0 + \bm{\alpha}_0) 
    \frac{ \Gamma\left( \frac{n + n_{01} + \delta_{01}}{2} \right) }{ (\tilde{s}_1^2 / 2)^{(n + n_{01} + \delta_{01})/2}} \left\lvert \bm{X}'\bm{X} + \bm{\Psi}_{01} \right\rvert^{-1/2} 
      \left[\prod_{k=2}^K \frac{\Gamma\left( \frac{n_{0k} + \delta_{0k}}{2} \right)}{{(\tilde{s}_{0k}^2/2)}^{(n_{0k} + \delta_{0k})/2}}
      \left\lvert \bm{\Psi}_{0k} \right\rvert^{-1/2}
      \right]}
      {
      \sum\limits_{\bm{c}_0 \in \{ 1, \ldots, K \}^{n_0}}
      B(\bm{n}_0 + \bm{\alpha}_0) 
    \frac{ \Gamma\left( \frac{n + n_{01} + \delta_{01}}{2} \right) }{ (\tilde{s}_1^2 / 2)^{(n + n_{01} + \delta_{01})/2}} \left\lvert \bm{X}'\bm{X} + \bm{\Psi}_{01} \right\rvert^{-1/2} 
      \left[\prod_{k=2}^K \frac{\Gamma\left( \frac{n_{0k} + \delta_{0k}}{2} \right)}{{(\tilde{s}_{0k}^2/2)}^{(n_{0k} + \delta_{0k})/2}}
      \left\lvert \bm{\Psi}_{0k} \right\rvert^{-1/2}
      \right]
      }  
      \label{eq:lm_margpost_c0}
\end{align}
is the marginal posterior mass function of $\bm{c}_0$, which is not computationally tractable in general. It follows that the marginal posterior of $\bm{\theta}_1$ is a mixture of $K^{n_0}$ normal-gamma densities, where the weights are determined by the mass function $f_H(\cdot)$.

Although it is computationally intractable, the form of $f_H(\cdot)$ provides insights into the borrowing properties of the LEAP. First, note that the mass function is decreasing in $\tilde{s}_1^2$. By an argument in completing the square, it can be shown that
$$\tilde{s}_1^2 = \bm{y}'\bm{y} + \tilde{\bm{\beta}}_{01}'\bm{\Psi}_{01}\tilde{\bm{\beta}}_{01}
- \left( \hat{\bm{\beta}}_1 - \tilde{\bm{\beta}}_{01} \right)'\left( \bm{X}'\bm{X} + \bm{\Psi}_{01} \right) \left( \hat{\bm{\beta}}_1 - \tilde{\bm{\beta}}_{01} \right),
$$
where $\left( \hat{\bm{\beta}}_1 - \tilde{\bm{\beta}}_{01} \right)'\left( \bm{X}'\bm{X} + \bm{\Psi}_{01} \right) \left( \hat{\bm{\beta}}_1 - \tilde{\bm{\beta}}_{01} \right) \ge 0 $ since $\bm{X}'\bm{X}$ and $\bm{\Psi}_{01}$ are positive definite matrices. Let $\lambda_{\text{min}} > 0$ be the minimal eigenvalue of $\left( \bm{X}'\bm{X} + \bm{\Psi}_{01} \right)$ and let $\lVert \cdot \rVert$ denote the Euclidean norm. Then, letting $\bm{\Omega}_{01} \to \bm{0}_{p\times p}$ so that $\bm{\Psi}_{01} \to \bm{X}_{01}'\bm{X}_{01}$, since
$$
\left( \hat{\bm{\beta}}_1 - \tilde{\bm{\beta}}_{01} \right)'
  \left( \bm{X}'\bm{X} + \bm{\Psi}_{01} \right)
  \left( \hat{\bm{\beta}}_1 - \tilde{\bm{\beta}}_{01} \right)
\ge 
\lambda_{\text{min}} \left\lVert \hat{\bm{\beta}}_1 - \tilde{\bm{\beta}}_{01} \right\rVert^2,
$$
the posterior mass of $\bm{c}_0$, all else being equal, is larger for partitions of $\{ 1, \ldots, K \}^{n_0}$ yielding
$\hat{\bm{\beta}}_{01} \approx \hat{\bm{\beta}}_1$, i.e., those partitions such that individuals in the first class result in an MLE that is similar to the current data set. This illustrates the ``adaptive borrowing'' of the LEAP, namely, the LEAP will borrow more from partitions of the historical data that are similar to the current data.

\subsection{Posterior under an improper prior}
We now show that, under a special class of improper initial priors, the LEAP yields a proper posterior. The improper prior is attractive as, under the fully exchangeable setting, it asymptotically mimics a frequentist analysis of pooling the data together.

Let the initial prior be given as
$$
    \pi_0(\bm{\theta}, \bm{\gamma}) \propto \tau_1^{\delta_{01}/2 - 1} \exp\left\{ -\frac{\tau_1\xi_{01}}{2} \right\}
    \left[\prod_{k=1}^K \gamma_k^{\alpha_{0k} - 1} \right]
    \prod_{k=2}^K 
    \phi_p(\bm{\beta}_k | \bm{\mu}_{0k}, \tau_k^{-1} \bm{\Omega}_{0k}^{-1})
    f_{\Gamma}\left( \tau \left| \frac{\delta_{0k}}{2}, \frac{\xi_{0k}}{2} \right) \right.
    .
$$
This prior is equivalent to taking $\bm{\Omega}_{01} \to \bm{0}_{p \times p}$ in the initial prior in \eqref{lm_initprior}. Clearly, this prior is not proper. However, the resulting posterior distribution is proper. As $\bm{\Omega}_{01} \to 0$, $\bm{\Psi}_{01} \to \bm{X}_0'\bm{C}_{01} \bm{X}_0$ and hence $\bm{\Lambda}_1 \to \left( \bm{X}'\bm{X} + \bm{X}_0'\bm{C}_{01} \bm{X}_0 \right)^{-1} \bm{X}_0'\bm{C}_{01} \bm{X}_0$. Also, $\bm{X}_{01}'\bm{C}_{01}\bm{X}_{01} = \sum_{i=1}^{n_0} c_{0i} \bm{x}_{0i} \bm{x}_{0i}' = \sum_{\{c_{0i} = 1\}} \bm{x}_{0i}\bm{x}_{0i}'$ is simply the precision matrix of the individuals classified to be exchangeable.

Hence, for a particular partition, the posterior mean is simply a weighted average of the MLE of the current data set and that for the partition of the historical data set (if it exists). That is, the mean for a component in the posterior in \eqref{eq:lm_posterior_theta1} is equal to the MLE of the current data set shrunken towards the MLE of the partition, with the mixture weight depending on the size of the partition and how similar the MLEs of the current and historical data sets are.

\section{Full data analysis results}

Patient characeristics of the ESTEEM I and ESTEEM II trials are provided in Table~\ref{tab:baseline}.

\newcolumntype{x}[1]{>{\centering\arraybackslash\hspace{0pt}}p{#1}}
\begin{small}
\begin{longtable}[h]{p{5.5cm}cx{2cm}cx{2cm}x{2.25cm}}
\toprule
 & \multicolumn{2}{c}{\textbf{ESTEEM I}} & \multicolumn{3}{c}{\textbf{ESTEEM II}} \\ 
\cmidrule(lr){2-3} \cmidrule(lr){4-6}
\textbf{Characteristic} & \textbf{N} & \textbf{Placebo} & \textbf{N} & \textbf{Placebo} & \textbf{30 mg BID}\\ 
\midrule
Age & 282 & 46.5 $\pm$ 12.7 & 411 & 45.7 $\pm$ 13.4 & 45.3 $\pm$ 13.1 \\ 
Smoker Category & 282 &  & 411 &  &  \\ 
    Current user &  & 92 (33\%) &  & 61 (45\%) & 101 (37\%) \\ 
    Not a current user &  & 190 (67\%) &  & 76 (55\%) & 173 (63\%) \\ 
Prior use of Systemic Therapies & 282 &  & 411 &  &  \\ 
    N &  & 132 (47\%) &  & 64 (47\%) & 117 (43\%) \\ 
    Y &  & 150 (53\%) &  & 73 (53\%) & 157 (57\%) \\ 
Baseline PASI Score & 282 & 19.4 (7.4) & 411 & 20.0 (8.0) & 18.9 (7.1) \\ 
\bottomrule
\caption{\label{tab:baseline} Patient characteristics and of ESTEEM I and ESTEEM II. Continuous covariates show Mean $\pm$ SD; binary covariates show $N$ ($\%$).
}
\end{longtable}
\end{small}
\vspace{2cm}

Table~\ref{tab:full_posterior_results} provides full posterior results for the ESTEEM 2 trial under the LEAPs, PBNPP, and reference prior, treating the ESTEEM 1 trial as historical data. The model is defined as follows
$$
  E(y | z, \bm{x}) = \beta_0 + \sum_{j=1}^6 \beta_j x_j,
$$
where $x_1$ is a treatment indicator, $x_2$ is centered and scaled age, $x_3 = x_2^2$, $x_4$ indicates that the subject does not currently smoke, $x_5$ is an indicator of whether the subject had received a prior systemic therapy, and $x_6$ is centered and scaled baseline PASI score.                    

    \begin{table}[h]
        \centering
        \resizebox{\textwidth}{!}{
            \begin{tabular}{lcccccccc}
            \toprule
             & \multicolumn{8}{c}{Prior} \\ \cmidrule(lr){2-9}
             & \multicolumn{2}{c}{LEAP $(K=2)$} & \multicolumn{2}{c}{LEAP $(K=3)$} & \multicolumn{2}{c}{PBNPP} & \multicolumn{2}{c}{Reference} \\ \cmidrule(lr){2-3}\cmidrule(lr){4-5}\cmidrule(lr){6-7}\cmidrule(lr){8-9}
            Parameter  & Post. Mean & Post. SD & Post. Mean & Post. SD & Post. Mean & Post. SD & Post. Mean & \multicolumn{1}{c}{Post. SD} \\ 
            \midrule
            $\beta_{0}$  & $-21.23$ & $3.915$ & $-21.37$ & $3.826$ & $-21.63$ & $3.755$ & $-17.97$ & $3.946$ \\
            $\beta_{1}$  & $-31.39$ & $3.353$ & $-31.49$ & $3.273$ & $-31.86$ & $3.085$ & $-32.12$ & $3.448$ \\
            $\beta_{2}$  & $\phantom{0}-1.06$ & $1.742$ & $\phantom{0}-1.13$ & $1.720$ & $\phantom{0}-0.73$ & $1.611$ & $\phantom{0}-0.48$ & $1.752$ \\
            $\beta_{3}$  & $\phantom{0}-1.98$ & $1.493$ & $\phantom{0}-2.08$ & $1.660$ & $\phantom{0}-1.65$ & $1.390$ & $\phantom{0}-2.29$ & $1.539$ \\
            $\beta_{4}$  & $\phantom{0}-2.23$ & $3.369$ & $\phantom{0}-2.04$ & $3.372$ & $\phantom{0}-2.23$ & $3.334$ & $\phantom{0}-4.32$ & $3.404$ \\
            $\beta_{5}$  & $\phantom{-}11.54$ & $3.245$ & $\phantom{-}11.93$ & $3.173$ & $\phantom{-}12.50$ & $3.216$ & $\phantom{-}10.11$ & $3.324$ \\
            $\beta_{6}$  & $\phantom{0}\phantom{-}1.48$ & $1.628$ & $\phantom{0}\phantom{-}1.48$ & $1.585$ & $\phantom{0}\phantom{-}1.38$ & $1.563$ & $\phantom{0}\phantom{-}1.29$ & $1.724$ \\
            $\sigma_1$  & $\phantom{-}35.44$ & $1.145$ & $\phantom{-}35.41$ & $1.112$ & $\phantom{-}35.27$ & $1.150$ & $\phantom{-}35.71$ & $1.237$ \\
            $\gamma_1$  & $\phantom{0}\phantom{-}0.40$ & $0.088$ & $\phantom{0}\phantom{-}0.41$ & $0.083$ & $\phantom{000}\phantom{-}$ & $\phantom{000}$ & $\phantom{000}\phantom{-}$ & $\phantom{000}$ \\
            $a_0$  & $\phantom{000}\phantom{-}$ & $\phantom{000}$ & $\phantom{000}\phantom{-}$ & $\phantom{000}$ & $\phantom{0}\phantom{-}0.34$ & $0.095$ & $\phantom{000}\phantom{-}$ & $\phantom{000}$ \\
            \bottomrule 
            \end{tabular}
        }
        \caption{Posterior means and standard deviations from the posterior densities of the ESTEEM-2 trial. The parameter $\sigma_1 = \tau_1^{-1/2}$ is the standard deviation of the outcome. Post. Mean = posterior mean; Post. SD = posterior standard deviation.}
        \label{tab:full_posterior_results}
    \end{table}

\begin{table}[ht]
    \centering
    \begin{tabular}{lcc}
        \toprule
        Parameter  & Post. Mean & \multicolumn{1}{c}{Post. SD} \\ 
        \midrule
        $\Delta$  & $-28.07$ & $2.69$ \\
        $\mu^{(0)}_{1}$  & $-13.74$ & $4.27$ \\
        $\mu^{(0)}_{2}$  & $\phantom{0}-8.05$ & $4.67$ \\
        $\mu^{(0)}_{3}$  & $-17.75$ & $4.35$ \\
        $\mu^{(0)}_{4}$  & $-12.74$ & $4.26$ \\
        $\mu^{(0)}_{5}$  & $-16.63$ & $3.67$ \\
        $\mu^{(1)}_{1}$  & $-37.02$ & $4.35$ \\
        $\mu^{(1)}_{2}$  & $-36.25$ & $4.85$ \\
        $\mu^{(1)}_{3}$  & $-41.76$ & $4.85$ \\
        $\mu^{(1)}_{4}$  & $-41.23$ & $3.88$ \\
        $\mu^{(1)}_{5}$  & $-52.98$ & $3.87$ \\
        $\sigma_{1}$  & $\phantom{-}33.84$ & $2.21$ \\
        $\sigma_{2}$  & $\phantom{-}38.58$ & $2.45$ \\
        $\sigma_{3}$  & $\phantom{-}35.13$ & $2.40$ \\
        $\sigma_{4}$  & $\phantom{-}34.80$ & $2.08$ \\
        $\sigma_{5}$  & $\phantom{-}29.14$ & $1.98$ \\
        \bottomrule 
    \end{tabular}
    \caption{Posterior summary statistics for the ESTEEM II trial under the PSIPP. $\Delta$ = treatment effect; $\mu_k^{(j)}$ = mean of $j^{th}$ arm in $k^{th}$ stratum; $\sigma_k$ = standard deviation of $k^{th}$ stratum.}
    \label{tab:my_label}
\end{table}
